%% file: Main.tex
\documentclass[11pt,a4paper]{article}

\usepackage[normalem]{ulem}
\usepackage{amsthm}
\usepackage{amsmath}
\usepackage{amssymb}
\usepackage{thmtools}
\usepackage{hyperref}
\usepackage{cleveref}
\usepackage{thm-restate}

\usepackage{enumitem}
\usepackage{xcolor}
\usepackage{amsfonts}
\usepackage{graphicx}
\usepackage{caption}
\usepackage{tikz}
\usepackage{xspace}
\usepackage{multicol}
\usepackage{svg}
\def\vertexnodes{\tikzstyle{every node}=[fill=black,circle, inner sep=2.2pt]}

\setlist[description]{font=\normalfont,leftmargin=\parindent,labelindent=\parindent}

\newtheorem{theorem}{Theorem}[section]

\newtheorem{lemma}[theorem]{Lemma}
\newtheorem{corollary}[theorem]{Corollary}
\newtheorem{conjecture}{Conjecture}

\newcommand{\sm}{\setminus}

\newcommand{\bpc}{\noindent {\em Proof of~(\theclaim). }}
\newcommand{\epc}{This proves~(\theclaim).}

\input{Macros}

\newcounter{claim}

\renewenvironment{proof}[1][]%
 {\noindent {\setcounter{claim}{0}\sc proof ---
   }{#1}{}}{\hfill$\Box$\vspace{2ex}} 

\newenvironment{claim}[1][]%
{\refstepcounter{claim}\vspace{1ex}\noindent{(\arabic{claim}) {#1}{}}\it}{\vspace{1ex}}

	{\noindent {}{#1}{}}{ This proves~(\arabic{claim}).\vspace{1ex}}

\newcommand*\samethanks[1][\value{footnote}]{\footnotemark[#1]}

\title{On the tree-width of even-hole-free graphs}
\date{}
\author{
	Pierre Aboulker
	\thanks{DIENS, {\'E}cole normale sup{\'e}rieure, CNRS, PSL University, Paris, France. 
	Email: pierreaboulker@gmail.com. \newline
	Supported by grant ANR-19-CE48-0016 from the French National Research Agency (ANR).}
	\and
	Isolde Adler
	\thanks{School of Computing, University of Leeds, Leeds, LS2 9JT, UK. 
	Email: i.m.adler@leeds.ac.uk. \newline This research was partly carried out during a visit to Universit{\'e} Paris Dauphine, funded by LAMSADE.} 
	\and
	Eun Jung Kim
	\thanks{Universit{\'e} Paris-Dauphine, PSL University, CNRS UMR7243, LAMSADE, Paris, France. Email: eun-jung.kim@dauphine.fr. \newline
	Supported by the grant from the French National Agency under JCJC program (ASSK: ANR-18-CE40-0025-01).
	}
	\and
	Ni Luh Dewi Sintiari
	\thanks{Univ Lyon, EnsL, UCBL, CNRS, LIP, F-69342, LYON Cedex 07, France.  
	Email: \{nicolas.trotignon,ni-luh-dewi.sintiari\}@ens-lyon.fr. \newline
	Partially
     supported by the LABEX MILYON (ANR-10-LABX-0070) of Universit\'e
     de Lyon, within the program `Investissements d'Avenir'     
	  (ANR-11-IDEX-0007) operated by the French National Research Agency
     (ANR).
	  }
	\and
	Nicolas Trotignon
	\samethanks
}

\begin{document}
\maketitle

\begin{abstract}
\input{Abstract}

\end{abstract}


\medskip\textbf{2012 ACM Subject Classification} Theory of computation $\rightarrow$ Mathematics of computing $\rightarrow$ Discrete mathematics $\rightarrow$ Graph Theory; 
Theory of computation $\rightarrow$ Design and analysis of algorithms $\rightarrow$ Streaming, sublinear and near linear time algorithms

\section{Introduction}
\label{sec:intro}
\input{Introduction}

\section{Notation}
\label{sec:notation}
\input{Preliminaries}

\section{Even-hole-free graphs excluding a minor}
\label{sec:minor}
\input{Minorfree}

\section{Subcubic (theta, prism)-free graphs}
\label{sec:subcubic}
\input{Subcubic}

\section{(Even hole, pyramid)-free graphs with maximum degree at most~4}
\label{sec:delta4}
\input{Delta4a}

\section{A possible structure theorem for even-hole-free graphs with maximum degree at most~4}
\label{sec:delta4b}
\input{Delta4b}

\bibliographystyle{plain}
\bibliography{Bibliography}

\end{document}

%% file: Macros.tex
\definecolor{C1}{RGB}{0,0,170}
\definecolor{C2}{RGB}{251,86,4}

\newcommand{\tw}{\operatorname{tw}}
\newcommand{\trigrid}[1]{\ensuremath{G^{\triangle}_{#1\times #1}}}

\newcommand{\dist}{\operatorname{dist}}

\newcommand{\N}{\mathbb{N}} 
\newcommand{\C}[1]{\mathcal{C}_{#1}} 
\newcommand{\G}{G} 

\renewcommand{\phi}{\varphi}
\newcommand{\ball}[2]{\ifx&#2& #1 \else (#1,#2)	\fi} 

\newcommand{\neighbourhood}[3]{N^{#1}_{#2}(#3)} 

%% file: Abstract.tex
The class of all even-hole-free graphs has unbounded tree-width, as it contains all complete graphs. Recently, a class of (even-hole, $K_4$)-free graphs was constructed, that still has unbounded tree-width [Sintiari and Trotignon, 2019]. The class has unbounded degree 
\emph{and} contains arbitrarily large clique-minors. We ask whether this is necessary.

We prove that for every graph $G$, if $G$ excludes a fixed graph $H$ as a minor, then 
	$G$ either has small tree-width, or
	$G$ contains a large wall or the line graph of a large wall as  
	\emph{induced} subgraph.
	This can be seen as a strengthening of Robertson and Seymour's excluded grid theorem 
	for the case of minor-free graphs.
	Our theorem implies that every class of even-hole-free graphs excluding a fixed graph as a minor has 
	bounded tree-width. In fact, our theorem applies to a more general class: (theta, prism)-free graphs. 
	This implies the known result 
	that planar even hole-free graph have bounded tree-width [da Silva and Linhares Sales, 
	Discrete Applied Mathematics 2010].

We conjecture that even-hole-free graphs of bounded degree have bounded tree-width. If true, this would mean that even-hole-freeness is testable in the bounded-degree graph model of property testing. We prove the conjecture for subcubic graphs and we give a bound on the tree-width of the class of (even hole, pyramid)-free graphs of degree at most $4$.

%% file: Introduction.tex
Here, all graphs are simple and undirected.  A \emph{hole} in a graph
is an induced cycle of length at least~4.  It is \emph{even} or
\emph{odd} according to the parity of its \emph{length}, that is the
number of its edges.  We say that a graph $G$ \emph{contains} a graph
$H$ if some induced subgraph of $G$ is isomorphic to $H$.  A graph is
\emph{$H$-free} if it does not contain $H$.  When $\cal H$ is a set of
graphs, $G$ is $\cal H$-free if $G$ contains no graph of $\cal H$. A
graph is therefore \emph{even-hole-free} if it does not contain an
even hole.

Even-hole-free graphs were the object of much attention, see for
instance the survey~\cite{vuskovic:evensurvey}.  However, many
questions about them remain unanswered, such as the existence of a
polynomial time algorithm to color them, or to find a maximum stable
set.  In fact, to the best of our knowledge, no problem that is
polynomial time solvable for chordal graphs is known to be NP-hard for
even-hole-free graphs (where a \emph{chordal} graph is a hole-free
graph). Despite the existence of several decomposition theorems or
structural properties (see~\cite{vuskovic:evensurvey}), no structure
theorem is known for even-hole-free graphs.

In addition, motivated by the question whether even-hole-freeness is testable
in the bounded degree model of property testing, the structure of
even-hole-free graphs of bounded maximum degree is of interest.  If
even-hole-free graphs of bounded degree have bounded tree-width, it
would imply testability in the bounded degree model, because
even-hole-freeness is expressible in monadic second-order logic with
modulo counting (CMSO) and CMSO is testable on bounded
tree-width~\cite{AdlerH18}. We will discuss this in greater detail below.
Let us first provide some more background.

\paragraph{Background.}
We begin by recalling known definitions and results about tree-width.  The
\emph{clique number} of a graph $G$, denoted by $\omega(G)$, is the
maximum number of pairwise adjacent vertices in $G$.  The \emph{tree-width} of
a graph $G$ is the minimum of $\omega(J)-1$ over all chordal graphs
$J$ such that $G$ is a subgraph of $J$. The tree-width can be seen as a measure of
the structural tameness of a graph: the smaller the tree-width, the more `tree-like' 
the graph. A celebrated
result~\cite{DBLP:journals/iandc/Courcelle90} asserts that many
problems (including graph coloring or finding a maximum stable set)
can be solved in polynomial time when restricted to graphs of bounded
tree-width.  However, many graphs with in some sense a simple structure
have large tree-width. For instance the complete graph on $n$
vertices, that we denote by $K_n$, has tree-width $n-1$.  
A graph $H$ is a \emph{minor} of a graph $G$, if $H$ can be obtained
from a subgraph of $G$ by contracting edges.
Tree-width
is monotone under taking minors in the sense that if $H$ is a minor
of $G$, then the tree-width of $H$ is less than or equal to the tree-width
of $G$.  It follows that graphs that contain $K_n$ as a minor have
tree-width at least $n-1$. The converse is not true:
grids have arbitrarily large tree-width but they do not contain
$K_5$ as a minor.

The class of all even-hole-free graphs trivially has unbounded tree-width, as is contains all
complete graphs.  Also, chordal graphs form a well studied subclass of
even-hole-free graphs of unbounded tree-width. However, even-hole-free
graphs with no triangle have bounded
tree-width~\cite{CameronSHV18}. This leads to asking
whether even-hole-free graphs of bounded clique number have bounded
tree-width -- a question that 
is first asked and motivated in~\cite{CameronCH18}. This was answered negatively
in~\cite{DBLP:journals/corr/abs-1906-10998}, where (even hole,
$K_4$)-free graphs of arbitrarily large tree-width are described.
However, the construction uses vertices of large degree \emph{and} a large
clique minor to increase the tree-width, and it seems natural to ask whether this
is necessary. 

It is known that planar even-hole-free graphs have bounded
tree-width~\cite{SilvaSS10}, and planar graphs do not contain $K_{\ell}$ as a minor for 
$\ell\geq 5$.
Besides that, it is known that
an upper bound on the length of the largest induced cycle implies an
upper bound on the tree-width for graphs of
bounded maximum degree~\cite{BodlaenderT97}.

\paragraph{Our contributions.}

The results explained above suggest the following two conjectures.
\begin{conjecture}
  \label{conj:minor}
  There is a function $f\colon\mathbb N\to\mathbb N$ such that every even-hole-free graph
	not containing $K_{\ell}$ as a minor
	has tree-width at most~$f(\ell)$.
\end{conjecture}

\begin{conjecture}
  \label{conj:degree}
	There is a function $f\colon\mathbb N\to\mathbb N$ such that every even-hole-free graph
   of degree at most $d$ has tree-width at most $f(d)$.
\end{conjecture}

In this paper we prove 
Conjecture~\ref{conj:minor} (cf.\ Section~\ref{sec:minor}). Indeed, we prove
the following stronger result, which implies Conjecture~\ref{conj:minor}.

\begin{restatable}[Induced grid theorem for minor-free graphs]{thm}{indwall}
	\label{thm:ind-wall}
	For every graph $H$ there is a function $f_{H}\colon\mathbb N\to\mathbb N$ such that every $H$-minor-free graph of tree-width at least $f_{H}(k)$ contains 
	 a $(k\times k$)-wall or the line graph of a chordless $(k\times k)$-wall
	as an \emph{induced} subgraph.
\end{restatable}

\noindent Here a \emph{wall} is a (possibly subdivided) hexagonal grid (cf.~Section~\ref{sec:minor}).

Slightly more generally, Theorem~\ref{thm:ind-wall} implies that \emph{(theta, prism)-free} graphs (to be defined in Section~2) exluding a fixed minor have bounded tree-width.
Note that (theta, prism)-free graphs form a superclass of even-hole-free graphs.

Our theorem can be seen as `induced' version on minor-free graphs classes 
of the following famous theorem.
\begin{restatable}[Robertson and Seymour~\cite{RobertsonS86}]{thm}{wall}
	\label{thm:wall}
	There is a function $f\colon \mathbb N\to \mathbb N$ such that every graph  of
	tree-width at least $f(k)$ contains a $(k\times k)$-wall as a subgraph.
\end{restatable}

Theorem~\ref{thm:wall} cannot be strengthened to finding walls as \emph{induced} subgraphs in general, 
because the complete graph $K_n$ has tree-width $n-1$ and only contains complete graphs as induced subgraphs.

Note that graphs with no induced subdivision of a $6\times 3$ wall, arbitrarily large tree-width and girth (in particular triangle and square-free) exist, as shown in~\cite{DBLP:journals/corr/abs-1906-10998}. But as for the construction of (even hole, $K_4$)-free graphs of unbounded tree-width in~\cite{DBLP:journals/corr/abs-1906-10998}, vertices of large degree and large clique minors are needed, we make the following conjecture. 
\begin{conjecture}
  \label{conj:wall-degree} For every $d\in \mathbb N$
	there is a function $f_d\colon\mathbb N\to\mathbb N$ such that every graph with 
	degree at most $d$ and tree-width at least $f_d(k)$ contains 
	a $(k\times k$)-wall or the line graph of a $(k\times k)$-wall
	as an \emph{induced} subgraph.
\end{conjecture}

Conjecture~\ref{conj:wall-degree} implies
Conjecture~\ref{conj:degree}. Conjecture~\ref{conj:wall-degree} is wide open, and our
results can be seen as a step in the direction of a proof.

For
Conjecture~\ref{conj:degree} (that is trivial for $l\leq 2$), we give
a proof for $l=3$ by providing a full structural description of
subcubic even-hole-free graphs (a graph is \emph{subcubic} if it does
not contain a vertex of degree more than~3).  In fact, all these
results apply to (theta, prism)-free graphs (cf.\ Section~\ref{sec:notation} for the details). We
also prove a weakening of Conjecture~\ref{conj:degree} for $l=4$ (cf.\
Section~\ref{sec:delta4}).

\paragraph{Motivation from property testing.}
Our other source of motivation for studying even-hole-free graphs of bounded degree 
stems from the question 
whether even-hole-freeness is testable in the bounded degree graph model.
Motivated by the growing need of highly efficient algorithms, in particular when the inputs are huge, property testing aims at devising sublinear time algorithms.
Property testing algorithms (simply called \emph{testers}) solve a relaxed version of decision problems, they
are randomised, and they come with a small controllable error probability.

Since the input cannot be read even once in sublinear time, 
 testers have local access to the
input graph only. The bounded degree graph model assumes a fixed
upper bound $d$ on the degree of all graphs, and the testers 
proceed by 
sampling a constant number of vertices of the input graph and exploring their 
local (constant radius) neighborhoods. A property $P$ is \emph{testable}, if there
is an $\epsilon$-tester for $P$, for every fixed small $\epsilon>0$. 
For an input graph $G$, 
an \emph{$\epsilon$-tester} determines, with probability at least $2/3$ 
correctly, 
whether $G$ has property $P$, or $G$ is $\epsilon$\emph{-far} from having 
property $P$. Here a \emph{property} is simply an isomorphism closed
class of graphs. A graph $G$ is $\epsilon$\emph{-close} to $P$,
if there is a graph $G'\in P$ on the same number $n$ of vertices as $G$,
such that $G$ and $G'$ can be made isomorphic by at most $\epsilon dn$ 
edge modifications (deletions or insertions) in $G$ or $G'$, and otherwise,
$G$ is $\epsilon$\emph{-far} from $P$.
The number of vertices explored in the input graph is called the \emph{query complexity}
of the tester. The model requires testers to have constant query complexity.

Properties that are known to be testable in the bounded degree graph model
include subgraph-freeness (for a fixed subgraph), $k$-edge connectivity, cycle-freeness, being Eulerian, degree-regularity~\cite{GoldreichRon2002}, bounded tree-width and minor-freeness~\cite{benjamini2010every,hassidim2009local,kumar2019random}, hyperfinite properties \cite{NewmanSohler2013}, $k$-vertex connectivity~\cite{yoshida2012property,forster2019computing}, and subdivision-freeness~\cite{kawarabayashi2013testing}. 
Properties that are \emph{not} testable in the bounded
degree model include bipartiteness, $3$-colorablilty, expansion properties, and $k$-clusterability, cf.~\cite{goldreich2017introduction}.

Since the testers can only explore a constant number of constant radius 
neighborhoods in the input graph, intuitively, properties that are testable should have some form of `local' nature. 
Hence it might seem unlikely that properties like Hamiltonicity and even-hole-freeness are testable, due to the large cycles involved. Indeed, it can be shown that there is no one-sided error tester for Hamiltonicity (i.\,e.\ there is no tester that always accepts yes-instances) with constant query complexity. More precisely, every one-sided tester has query complexity at least $\Omega(n)$~\cite{AdlerKoehler2020}.

Perhaps surprisingly, our results suggest a different picture for 
even-hole-freeness. It is known that on graphs of 
bounded degree and bounded tree-width, every property that can be expressed 
in monadic second-order logic with counting (CMSO), is testable with constant query complexity and polylogarithmic 
running time~\cite{AdlerH18}. (Here \emph{polylogarithmic in} $n$ means bounded by a polynomial in $\log n$.) 
It is straightforward to see that even-hole-freeness is expressible in CMSO. This can be done by expressing that there is no set $X$ of edges that form an induced hole, where $|X|$ is even (cf.\ eg.\ \cite{CourcelleE12} for more details on CMSO expressibility).
Together with the fact that bounded tree-width is testable, this implies that 
\emph{if} Conjecture~\ref{conj:degree} is true, 
then even-hole-freeness is testable
with constant query complexity and polylogarithmic 
running time. (This is done by first testing for bounded tree-width and, if the answer is positive, testing for even-hole-freeness using CMSO testability.)
Our results in Section~\ref{sec:subcubic} 
imply the following.
\smallskip

\begin{restatable}{thm}{subcubic-testability}
On subcubic graphs, even-hole-freeness is testable
with constant query complexity and polylogarithmic 
running time.
\end{restatable}

\paragraph{Structure of the paper.} We start with fixing notation in Section~\ref{sec:notation}.
Section~\ref{sec:minor} contains the proof of the induced grid theorem for minor-free
graph classes, and the proof of Conjecture~\ref{conj:minor}.
Section~\ref{sec:subcubic} contains the structure theorem for subcubic (theta, prism)-free graphs and the proof of Conjecture~\ref{conj:degree} for $d=3$. 
Recall that (theta, prism)-free graphs are defined in the next section, and it is a superclass of even-hole-free graphs.
In Section~\ref{sec:delta4} we
provide a structure theorem for (even hole, pyramid)-free graphs of maximum degree $4$, and
we derive a bound on the tree-width of this class (pyramids will be defined in the next section). In Section~\ref{sec:delta4b} we give ideas suggesting that a structure theorem for even-hole-free graphs with maximum degree~4 might exist, and if so, should imply bounded tree-width.

%% file: Preliminaries.tex
We let $\N$ denote the set of natural numbers including $0$. 
We use $X\sqcup Y$  instead of $X\cup Y$ if $X\cap Y=\emptyset$. We use $(0,1]$ to denote the real
interval that excludes $0$ and includes $1$.
For any $n\in \N, n\geq 1$, let $[n]:=\{1,\dots,n\}$.

An \emph{(undirected) graph} is a pair $G=(V(G),E(G))$, consisting of a set $V(G)$, the set of \emph{vertices} of $G$, and a set $E(G)$ of \emph{edges} of $G$, where an edge is a two-vertex subset of $V(G)$. A graph $H$ is a \emph{subgraph} of a graph $G$, if $V(H)\subseteq V(G)$ and $E(H)\subseteq E(G)$. For a set $X\subseteq V(G)$ the subgraph \emph{induced} by
$X$ in $G$ is the subgraph $G[X]$ of $G$ with vertex set $X$, such that $e\in E(G[X])$ iff $e\in E(G)$ and $e\subseteq X$. A graph $H$ is an \emph{induced subgraph} of $G$, if $H=G[X]$ for some $X\subseteq V(G)$. 
For a set $S\subseteq V(G)$ we let $G\setminus S:=G[V(G)\setminus S]$ and if $S=\{v\}$ is a singleton set, then we write $G\setminus v$ instead of $G\setminus \{v\}$.
%

A {\em path} in $G$ is a sequence $P$ of distinct vertices
$p_1 \dots p_n$, where for $i,j \in \{1,\dots,n\}$,
$p_ip_{j} \in E(G)$ if and only if $|i-j|= 1$.  For two vertices
$p_i, p_j \in V(P)$ with $j > i$, the path $p_ip_{i+1}\dots p_j$ is
a {\em subpath of $P$} that is denoted by $p_iPp_j$. The subpath
$p_2\dots p_{n-1}$ is called the {\em interior} of $P$. 
The vertices
$p_1,p_n$ are the \emph{ends} of the path, and the vertices in the
interior of $P$ are called the \emph{internal} vertices of $P$.  
A {\em cycle} is defined similarly, with the additional properties that $n\geq 4$ and $p_1 = p_n$.  The {\em length} of a path $P$ is the number of edges of $P$.  The length of cycle is defined similarly.
Let $G$ be a graph. For vertices $u,v\in V(G)$, the \emph{distance} between $u$ and $v$ in $G$, 
denoted by $\dist_G(u,v)$, is the length of a shortest path from $u$ to $v$, if a path exists,
and $\infty$ otherwise. For two subsets $X,Y\subseteq V(G)$, the \emph{distance} between $X$ and $Y$ is $\min\{\dist_G(x,y)\mid x\in X, y\in Y\}$.
For $v\in V(G)$, we call the set $\neighbourhood{\G}{r}{v}:=\{w\in V(\G)\mid \dist_G(v,w)\leq r\}$ the \emph{$r$-neighborhood} of $v$ (in $G$). 

The \emph{degree} of a vertex $v$ in $G$ is defined as $\deg_\G(v):=|\{u\in V(G)\mid \{u,v\}\in E(G)\}|$. The \emph{degree} of $G$, $\deg(G)$, is the maximum degree over all vertices of $G$.
By $\C{d}$ we denote the class of all graphs of degree at most $d\in \N$.

The \emph{line graph} of a graph $G$ is the graph $L(G)$, with $V(L(G))=E(G)$ and two vertices of $L(G)$ are
adjacent, if their corresponding edges are incident in $G$.
An edge $e\in E(G)$ is a \emph{chord}
of cycle $C$, if the endpoints of $e$ are vertices of $C$ that are not adjacent on $C$.  
A \emph{hole} is a chordless cycle of length at least~4.
A \emph{clique} in $G$ is a set $X\subseteq V(G)$ of vertices such that $\{v,w\}\in E(G)$ for every pair $v,w\in X$ with $v\neq w$. A graph $K$ is \emph{complete}, if $V(K)$ is a clique in $K$. 
We use $K_k$ to denote the complete graph on $k$ vertices.
For disjoint sets $A,B \subseteq V(G)$, we say that $A$ is {\em
  anticomplete to} $B$ if no edges are present between $A$ and $B$ in
$G$.

%
%
%

A {\em pyramid} is a graph made of three chordless paths
$P_1 = x \dots a$, $P_2 = x \dots b$, $P_3 = x \dots c$, each of
length at least~1, two of which have length at least 2, internally
vertex-disjoint, and such that $abc$ is a triangle and no edges exist
between the paths except those of the triangle and the three edges
incident to~$x$.  The vertex $x$ is called the {\em apex} of the
pyramid.

A \emph{prism} is a graph made of three vertex-disjoint chordless
paths $P_1 = a \dots a'$, $P_2 = b \dots b'$, $P_3 = c \dots c'$ of
length at least 1, such that $abc$ and $a'b'c'$ are triangles and no
edges exist between the paths except those of the two triangles.

A \emph{theta} is a graph made of three internally vertex-disjoint
chordless paths $P_1 = a \dots b$, $P_2 = a \dots b$,
$P_3 = a \dots b$ of length at least~2 and such that no edges exist
between the paths except the three edges incident to~$a$ and the three
edges incident to~$b$.

A {\em wheel} is a graph formed from a hole $H$ together with a vertex
$x$ that has at least three neighbors in the hole.  Such a hole $H$ is
called the {\em rim}, and such a vertex $x$ is called the {\em center}
of the wheel.  We denote by $(H,x)$, the wheel with rim $H$ and center
$x$.

\begin{figure}[ht]
  \begin{center}
    \includegraphics[height=2cm]{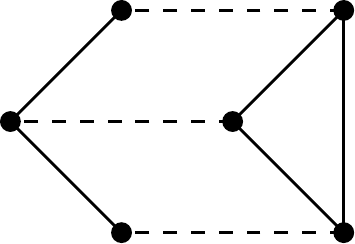} \hspace{2ex}
    \hspace{.2em}
    \includegraphics[height=2cm]{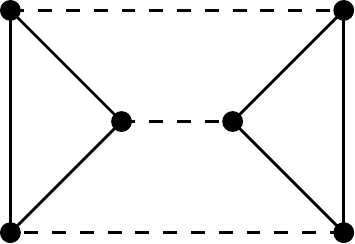} \hspace{2ex}
    \hspace{.2em}
    \includegraphics[height=2cm]{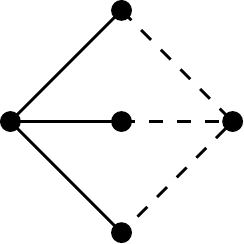} \hspace{2ex}
    \hspace{.2em}
    \includegraphics[height=2cm]{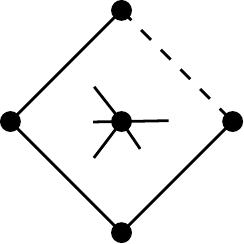}
  \end{center}
  \caption{Pyramid, prism, theta, and wheel (dashed lines represent
    paths)\label{f:tc}} 
\end{figure}

Theta and prism are relevant in this work, because of the following 
well-known lemma. The following lemma clearly implies that
(theta, prism)-free graphs form a superclass of
even-hole-free graphs.

\begin{lemma}
\label{lem:ehf-theta-prism}
  Every theta and every prism contains an even hole.
\end{lemma}

\begin{proof}
It follows by the fact that there exist two paths
in a theta or in a prism that have same parity,
which induce an even hole.
\end{proof}

%% file: Minorfree.tex
In this section we prove an `induced grid theorem' for graphs excluding a fixed minor. From this we derive that even-hole-free graphs excluding a fixed minor have bounded tree-width.

We begin by defining grids and walls.
Let $n,m$ be integers with $n,m\geq 2$. An \emph{$(n\times m)$-grid} is 
the graph $G_{n\times m}$ with $V(G_{n\times m})=[n]\times [m]$ and 
\[E(G_{n\times m})=\big\{\{(i_1,j_1),(i_2,j_2)\}\mid |i_1-i_2|+ |j_1-j_2|=1,i_1,i_2\in [n],j_1,j_2\in [m]\big\}.\]
Figure~\ref{fig:grid} shows $G_{5\times 5}$.

\begin{figure}[h]
\centering
	\newcommand{\Num}{4} 
	\newcommand{\Mum}{4} 
\begin{tikzpicture}[scale=.6, v/.style= {circle,draw,fill, scale=.2pt}, ]
                        foo/.style= {draw,circle,inner sep=2pt,fill}]

\foreach \j in {0,...,4}{%
         \foreach \i in {0,...,4}{
				\node[v] (h\j;\i) at ({\j},{\i}) {};
			}
}
\foreach \j in {0,...,3}{%
         \foreach \i in {0,...,\Num}{
				\draw (h\j;\i) -- ({\j+1},{\i}) ;
			}
}

\foreach \j in {0,...,\Mum}{%
         \foreach \i in {0,...,3}{
				\draw (h\j;\i) -- ({\j},{\i+1}) ;
			}
}

\end{tikzpicture}

	\caption{The $(5\times 5)$-grid $G_{5\times 5}$.}
\label{fig:grid}
\end{figure}
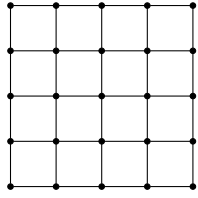

Let $n,m \geq 2$ be integers. An \emph{elementary $(n\times m)$-wall} is a graph $G=(V,E)$ with vertex set
\begin{align*}V=&\big\{ (1,2j-1)\mid 1\leq j\leq m\big\} \cup 
	\big\{ (i,j)\mid 1 < i < n, 1 \leq j\leq 2m\big\}\\
	&\cup \big\{ (n,2j-1)\mid 1 \leq j\leq m, \text{ if $n$ is even}\big\} 
	\cup \big\{ (n,2j)\mid 1 \leq j\leq m, \text{ if $n$ is odd}\big\}
\end{align*}
and edge set
\begin{align*}E=& \big\{ (1,2j-1),(1,2j+1)\mid 1\leq j\leq m-1\big\} \cup
	\big\{ \{(i,j),(i,j+1)\}\mid 2\leq i< n,1\leq j< 2m\big\}\\
	 &\cup\big\{ \{(n,2j),(n,2j+2))\}\mid 1\leq j < m\text{ if $n$ is odd}\big\}\\ 
	 &\cup\big\{ \{(n,2j-1),(n,2j+1)\}\mid 1\leq j < m\text{ if $n$ is even}\big\}\\ 
	 &\cup\big\{ \{(i,j),(i+1,j)\}\mid 1\leq i<n, 1\leq j \leq 2m,\text{ $i,j$ odd}\big\}\\
	 &\cup\big\{\{ (i,j),(i+1,j)\}\mid 1\leq i< n, 1\leq j \leq 2m,\text{ $i,j$ even}\big\}. 
\end{align*}
Figure~\ref{fig:elementarywall} shows an elementary $(5\times 5)$-wall.

\begin{figure}
\centering
	\newcommand{\Num}{3} 
	\newcommand{\Mum}{8} 
\begin{tikzpicture}[scale=.6, v/.style= {circle,draw,fill, scale=.2pt}, ]
                        foo/.style= {draw,circle,inner sep=2pt,fill}]

\foreach \j in {0,...,\Mum}{%
         \foreach \i in {0,...,2}{
				\node[v] (h\j;\i) at ({\j},{\i}) {}; 
			}
}
\foreach \j in {0,2,...,\Mum}{%
				\node[v] (h\j;\Num) at ({\j},{\Num}) {}; 
				\node[v] (hh\j;-1) at ({\j+1},{-1}) {}; 
}

\foreach \j in {0,...,\Mum}{%
         \foreach \i in {0,...,2}{
				\draw ({\j},{\i}) -- ({\j+1},{\i}) ; 
			}
}
\foreach \j in {0,...,7}{%
				\draw ({\j},{\Num}) -- ({\j+1},{\Num}) ; 
}

	\foreach \j in {0,...,7}{%
				\draw ({\j+2},{-1}) -- ({\j+1},{-1}) ; 
}

\foreach \j in {0,2,...,\Mum}{%
			\foreach \i in {0,2,...,\Num}{
				\draw ({\j},{\i}) -- ({\j},{\i+1}) ;  
				\draw ({\j+1},{\i}) -- ({\j+1},{\i-1}) ; 
			}
    }
\end{tikzpicture}

	\caption{The elementary $(5\times 5)$-wall.}
\label{fig:elementarywall}
\end{figure}
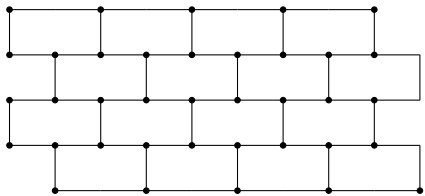

An elementary $(n\times m)$-wall has $n$ \emph{horizontal} paths, where
the first  horizontal path is induced by the vertex set 
$\big\{ (1,2j-1)\mid 1\leq j\leq 2m\big \}$,
the $i$th horizontal path is induced by the vertex set 
$\big\{ (i,j)\mid 1\leq j\leq 2m\big \}$. for $1 < i<n$, and the
$n$th horizontal path is induced by the vertex set
$\big\{ (n,2j-1)\mid 1\leq j\leq 2m\big \}$ if $n$ is odd, and by
$\big\{ (n,2j)\mid 1\leq j\leq 2m\big \}$ if $n$ is even.
An elementary $(n\times m)$-wall has $m$ \emph{vertical} paths, where
the $j$th vertical path is induced by the vertex set 
$\big\{ (i,2j-1),(i+1,2j-1)\mid 1\leq i< n, \text{ $i$ odd}\big\}\cup
\big\{ (i,2j),(i+1,2j)\mid 1\leq i< n, \text{ $i$ even}\big\}$,
for $1\leq j \leq m$. Figure~\ref{fig:vertical-paths} shows an elementary
wall with vertical paths.

\begin{figure}
\centering
	\newcommand{\Num}{3} 
	\newcommand{\Mum}{8} 

\begin{tikzpicture}[scale=.6, v/.style={circle,draw,fill,scale=.2pt}, ]
                        foo/.style={draw,circle,inner sep=2pt,fill}]%
\foreach \j in {0,...,8}{%
         \foreach \i in {0,...,2}{%
				\node[v] (h\j;\i) at ({\j},{\i}) {}; 
			}%
}%
\foreach \j in {0,2,...,\Mum}{%
				\node[v] (h\j;\Num) at ({\j},{\Num}) {}; 
				\node[v] (hh\j;-1) at ({\j+1},{-1}) {}; 
}%
\foreach \j in {0,...,\Mum}{%
         \foreach \i in {0,...,2}{%
				\draw ({\j},{\i}) -- ({\j+1},{\i}) ; 
			}%
}%
\foreach \j in {0,...,7}{%
				\draw ({\j},{\Num}) -- ({\j+1},{\Num}) ; 
}%
	\foreach \j in {0,...,7}{%
				\draw ({\j+2},{-1}) -- ({\j+1},{-1}) ; 
}%
\foreach \j in {0,2,...,\Mum}{%
			\foreach \i in {0,2,...,\Num}{%
				\draw ({\j},{\i}) -- ({\j},{\i+1}) ;  
				\draw ({\j+1},{\i}) -- ({\j+1},{\i-1}) ; 
			}%
    }%
	 \draw[line width=3pt] (0,3) -- (0,2) ; 
	 \draw[line width=3pt] (0,1) -- (0,0) ; 
	 \draw[line width=3pt] (1,2) -- (1,1) ; 
	 \draw[line width=3pt] (1,0) -- (1,-1) ; 
	 \draw[line width=3pt] (0,0) -- (1,0) ; 
	 \draw[line width=3pt] (0,1) -- (1,1) ; 
	 \draw[line width=3pt] (0,2) -- (1,2) ; 
	 \draw[line width=3pt] (4,3) -- (4,2) ; 
	 \draw[line width=3pt] (4,1) -- (4,0) ; 
	 \draw[line width=3pt] (5,2) -- (5,1) ; 
	 \draw[line width=3pt] (5,0) -- (5,-1) ; 
	 \draw[line width=3pt] (4,0) -- (5,0) ; 
	 \draw[line width=3pt] (4,1) -- (5,1) ; 
	 \draw[line width=3pt] (4,2) -- (5,2) ; 
	 \draw[line width=3pt] (8,3) -- (8,2) ; 
	 \draw[line width=3pt] (8,1) -- (8,0) ; 
	 \draw[line width=3pt] (9,2) -- (9,1) ; 
	 \draw[line width=3pt] (9,0) -- (9,-1) ; 
	 \draw[line width=3pt] (8,0) -- (9,0) ; 
	 \draw[line width=3pt] (8,1) -- (9,1) ; 
	 \draw[line width=3pt] (8,2) -- (9,2) ; 
	 
\end{tikzpicture}
\caption{An elementary $(5\times 5)$-wall with three of the five vertical paths highlighted.}
\label{fig:vertical-paths}
\end{figure}
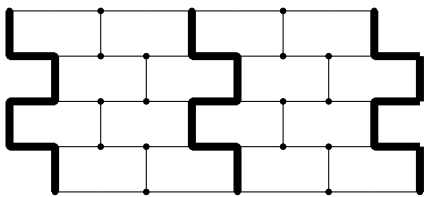

An \emph{$(n\times m)$-wall} is a subdivision of an elementary $(n\times m)$-wall.
An \emph{$(n\times m)$-wall} also has $n$ horizontal paths and $m$ vertical paths, which arise from
the paths of the underlying elementary wall including the subdivision vertices.

Robertson and Seymour's famous excluded grid theorem states that for any graph $G$, either $G$ has small tree-width or $G$ contains a large grid minor, thus identifying large grid minors as canonical obstructions to small tree-width. 

\begin{theorem}[Robertson and Seymour~\cite{RobertsonS86}]\label{thm:grid}
	There is a function $f$ such that for every $k\geq 1$ and every graph $G$, 
	if $\tw(G)\geq f(k)$ then $G$ contains the $(k\times k)$-grid as a minor.
\end{theorem}

The theorem continues to hold if we replace \emph{grid} by \emph{wall}, and the latter have the advantage of having a maximum degree of three, which allows finding walls as \emph{subgraphs}.
Recall Theorem~\ref{thm:wall}.
	\wall*

Obviously, this cannot be strengthened to finding walls as \emph{induced} subgraphs, 
because the complete graph $K_n$ has tree-width $n-1$ and only contains complete graphs as induced subgraphs.
A graph $G$ is called \emph{chordless}, if no cycle of length at least $4$ in $G$ 
has a chord in $G$.
For graphs excluding a fixed minor, we now prove an `induced grid theorem'.
Recall Theorem~\ref{thm:ind-wall}.
	\indwall*

Observe that the line-graph of a chordless wall has degree at most $3$.
For the proof of Theorem~\ref{thm:ind-wall} we need some lemmas and notation.
A \emph{fork} is a tree with exactly three leaves.  
A \emph{semi-fork} is a graph obtained from a triangle by appending disjoint 
paths of length at least $1$ at each vertex of the triangle. 
Note that both a fork and a semi-fork have precisely three degree-one vertices.

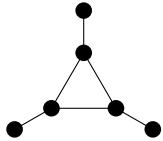
\begin{figure}[h]
\begin{center}	
\begin{tikzpicture}
\begin{scope}[scale=.7]
	\vertexnodes 
	\coordinate(top) at (90:0.7cm);
	\coordinate(left) at (210:0.7cm);
	\coordinate(right) at (-30:0.7cm);
	\coordinate(extop) at (90:1.5cm);
	\coordinate(exleft) at (210:1.5cm);
	\coordinate(exright) at (-30:1.5cm);
	\draw  (top) node{} -- (left) node{} -- (right) node{} -- cycle;
	\draw (top) -- (extop) node{};
	\draw (left) -- (exleft) node{};
	\draw (right) -- (exright) node{};
\end{scope}
\end{tikzpicture}
\end{center}
\caption{The net graph.}\label{fig:net-graph}
\end{figure}
A semi-fork obtained from a triangle by  appending disjoint 
paths of length \emph{exactly} $1$ at each vertex of the triangle is called a \emph{net graph} (cf.\ Figure~\ref{fig:net-graph}).
Let $G$ be a graph and let $v\in V(G)$ be a vertex of degree $3$ in $G$ with neighbors $a,b,c$. 
The graph obtained by a \emph{net graph replacement at} $v$ is the graph $H$, with
$V(H)=(V(G)\setminus \{v\})\cup \{x,y,z\}$, where $x,y,z$ are new vertices, i.\,e.\ 
$x,y,z\notin V(G)$, and $E(H)=E(G\setminus v)\cup\{xy,yz,xz\}\cup \{xa,yb,zc\}$.

The `grid-like' configurations we aim for are walls and line graphs of walls. On the way, we also encounter 
a slightly untidier `mix' of both, that we call stone wall. An \emph{$(n\times m)$-stone wall}
a graph obtained from a wall $W$ by picking a (possibly empty) subset $X$ of the degree-$3$-vertices of 
$W$ and performing net graph replacements at each vertex in $X$.
Note that if we perform net graph replacements at all degree-$3$-vertices of the wall $W$, the resulting
graph is the line graph of a wall (namely of the wall $W'$ obtained from $W$ by adding an additional
subdivision vertex on each path that connects two degree-$3$-vertices). Observe that
stone walls have maximum degree $3$.
A stone wall is \emph{homogeneous}, if it is either a wall or the line graph of a wall. Note that because stone walls have degree at most $3$, 
if a homogeneous stone 
wall is the line graph of a wall $W$, $W$ must be chordless.

A \emph{triangulated $(n\times m)$-grid} is the graph $G^{\triangle}_{n\times m}$ with $V(G^{\triangle}_{n\times m})=V(G_{n\times m})$ and
$E(G^{\triangle}_{n\times m})=E(G_{n\times m}) \cup \big\{ \{(i,j),(i-1,j+1)\}\mid 1<i\leq m, 1\leq j<m \big\}$, see Figure~\ref{fig:trigrid}.
The vertices $(1,1)$, $(1,m)$, $(n,1)$ and $(n,m)$ are called the \emph{corners} of the (triangulated) grid.


\begin{figure}
\centering
	\newcommand{\Num}{4} 
	\newcommand{\Mum}{4} 
\begin{tikzpicture}[scale=.6, v/.style= {circle,draw,fill, scale=.2pt}, ]
                        foo/.style= {draw,circle,inner sep=2pt,fill}]

\foreach \j in {0,...,\Mum}{%
         \foreach \i in {0,...,\Num}{
				\node[v] (h\j;\i) at ({\j},{\i}) {};
			}
}
\foreach \j in {0,...,3}{%
         \foreach \i in {0,...,\Num}{
				\draw (h\j;\i) -- ({\j+1},{\i}) ;
			}
}

\foreach \j in {0,...,\Mum}{%
         \foreach \i in {0,...,3}{
				\draw (h\j;\i) -- ({\j},{\i+1}) ;
			}
}
\foreach \j in {0,...,3}{%
         \foreach \i in {0,...,3}{
				\draw (h\j;\i) -- ({\j+1},{\i+1}) ;
			}
}

\end{tikzpicture}
	\caption{The triangulated $(5\times 5)$-grid $\trigrid{5}$.}
	\label{fig:trigrid}
\end{figure}
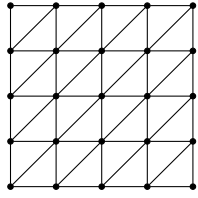

The following lemma shows that a large wall contains many smaller induced subwalls.

\begin{lemma}\label{lem:subwalls}
	For $n,m\in \mathbb N$, let $W$ be an $(n\times m)$-wall.
	Let $X$ be a set of pairwise non-adjacent rows of $W$ and let $Y$ be a set of 
	pairwise non-adjacent columns of $W$.  
	Then $W$ contains an induced $(|X|\times |Y|)$-wall $W'$, 
	with $V(W')\subseteq \bigcup_{P\in X}V(P) \cup \bigcup_{Q\in X}V(Q)$.
\end{lemma}
\begin{proof}
	Assume that $|X|\geq 2$ and $|Y|\geq 2$. We obtain $W'$ by taking the subgraph of $W$ induced by the set $\bigcup_{P\in X}V(P) \cup \bigcup_{Q\in X}V(Q)$
	and repeatedly deleting degree-$1$-vertices until all vertices have degree at least $2$.
\end{proof}


For the proof of Theorem~\ref{thm:ind-wall}, we use a corollary of the main result of \cite{FominGT11}, and we need the notion of \emph{contraction}.
Let $G$ and $H$ be graphs. If $H$ can be obtained from $G$ by a
sequence of edge contractions, then $H$ is called a \emph{contraction}
of $G$.  Alternatively, contractions can be defined via mappings as
follows.  Let $G$ and $H$ be graphs and let $\phi\colon V(G)\to V(H)$
be a surjective mapping such that
\begin{enumerate}
\item for every vertex $v\in V(H)$, its pre-image $\phi^{-1}(v)$ is
  connected in $G$,
\item for every edge $uv\in E(H)$, the graph
  $G[\phi^{-1}(u)\cup\phi^{-1}(v)]$ is connected,
\item for every edge $xy\in E(G)$, either $\phi(u)=\phi(v)$ or
  $\phi(u)\phi(v)\in E(H)$.
\end{enumerate}

\begin{corollary}[Fomin, Golovach and Thilikos \cite{FominGT11}]\label{cor:contractions}
Let $H$ be a graph and let $G$ be a graph excluding $H$ as a minor.
There exists a constant $c_H$ such that if $\tw(G) \geq c_H \cdot
	(k+1)^2$, then $G$ contains an induced subgraph 
	that contains $\trigrid{k}$ as a contraction.
\end{corollary}
%
%
%

The following lemma will help us to find a large stone wall in a graph containing a 
large triangulated grid as a contraction.

\begin{lemma}\label{lem:forks}
Let $G$ be a connected graph whose vertex set is
  partitioned into connected sets $A$, $A'$, $B$, $B'$, $C$, $C'$ and
  $S$.  Suppose that every edge of $G$ has either both ends in one of
  the sets, or is from $A'$ to $A$, from $B'$ to $B$, from $C'$ to $C$,
  or from $S$ to $A\cup B \cup C$.

    If $a\in A'$, $b\in B'$ and $c\in C'$, then $a$, $b$ and $c$ are
    the degree one vertices of some induced fork or semi-fork of $G$.
\end{lemma}
\begin{proof}
	Let $P$ be a shortest path from $b$ to $c$ in
    $G[B'\cup B \cup S \cup C \cup C']$.  Note that $P$ must go
    through $S$.  Let $Q = a \dots w$ in $G[A'\cup A\cup S]$ be a
    shortest path such that $w$ has some neighbors in $P$.  Let $u$
    (resp.\ $v$) be the neighbor of $w$ in $P$ closest to $b$ (resp.\
    to $c$) along $P$.  Note that $u\neq b$ and $v\neq c$.  If $u=v$,
    then $P$ and $Q$ form a fork. If $u$ is adjacent to $v$, then $P$
    and $Q$ form a semi-fork. If $u\neq v$ and $uv\notin E(G)$, then
    $aPw$, $bQu$ and $cQv$ form a fork.  In all cases, $a$, $b$ and
    $c$ are the three vertices of degree 1 of the fork or semi-fork. 
\end{proof}

For tidying up stone walls, we make use of a natural variant of Ramsey's Theorem for bipartite graphs, first introduced by Beineke and Schwenk in 1975.

\begin{theorem}[Beineke and Schwenk~\cite{BeinekeS75}]\label{thm:bipartite-ramsey}
	For every integer $r\geq 1$ there exists a smallest positive integer $n=n(r)$, such that any 2-edge-coloring of the complete 
	bipartite graph $K_{n,n}$ contains a monochromatic $K_{r,r}$.
\end{theorem}

In~\cite{Thomason82} it was shown that $n(r)\leq 2^r(r-1)+1$. 

The next lemma shows that any sufficiently large stone wall also contains a large homogeneous
stone wall as induced subgraph.

\begin{lemma}\label{lem:tidy-stonewall}
	For every integer $r\geq 2$ there exists an integer $n=n(r)$ such that every $(n\times n)$-stone wall contains a homogeneous $(r\times r)$-stone wall as induced subgraph.
\end{lemma}

\begin{proof}
 
	Given $r$, let $n=n(r)$ be large enough. Given an $(n\times n)$-stone wall $W$, we define
	an auxiliary wall $W'$, which is obtained from $W$ by contracting every triangle. Each vertex
	in $W'$ that is the result of contracting a triangle is colored red (red encodes `semi-fork'), 
	and all other degree-$3$-vertices of $W'$ are colored green (green encodes `fork').

   Define a complete bipartite graph $H$ with $V(H)=A \cup B$ as follows. 
	The elements of $A$ are horizontal paths of $W'$, and the elements of $B$ are
	the vertical paths in $W'$.
	Note that each vertical path has two colored vertices in common with each horizontal path.
 
 We fix an orientation of the horizontal paths `from left to right'.
	Now we color the edges of $H$ with four colors.
	Let $P\in A$ be a horizontal path and let $Q\in B$ be a vertical path.\\
	1. If $V(P)\cap V(Q)\subseteq V(W')$ contains two green vertices, 
	 we color the edge $PQ$ green.\\
	2. If $V(P)\cap V(Q)\subseteq V(W')$ contains two red vertices, 
	 we color the  edge $PQ$ red.\\
	3. If $V(P)\cap V(Q)\subseteq V(W')$ contains a green and a red vertex,
	and the green vertex appears before the red vertex when traversing $P$
	from left to right, then
	 we color the edge $PQ$ white.\\
	4. If $V(P)\cap V(Q)\subseteq V(W')$ contains a green and a red vertex,
	and the red vertex appears before the green vertex when traversing $P$
	from left to right, then
	 we color the edge $PQ$ black.\\

	By applying Theorem~\ref{thm:bipartite-ramsey} (multiple times, if necessary), we find that 
	$H$ contains a large monochromatic 
	complete bipartite subgraph $H'$.
 
   If $H'$ is green (or red, respectively), 
	we find a large subwall in $W'$ where all vertices of degree $3$ 
	are green (red, respectively) as follows. We take the horizontal and vertical 
	paths in $W'$ that correspond to $V(H')$, 
	leaving out every second path to make sure that the horizontal paths we 
	keep are pairwise non-adjacent, and that the vertical paths we 
	keep are pairwise non-adjacent.
	Then we apply Lemma~\ref{lem:subwalls}.
	Undoing the contractions of triangles in the case that $H'$ is red, we thus obtain a large induced 
	homogeneous stone wall in~$W$.
 
	In the case that $H'$ is white or black, we find a large subwall $W''$ in $W'$ where both red and green appear at each intersection of a horizontal and a vertical path. W.\,l.\,o.\,g.\ assume that $H'$ is white (otherwise flip the wall exchanging left and right). We will now explain how to  
	find a large induced subwall of~$W$. 
	
	Let $X$ be a maximal subset of horizontal paths of $W''$ of pairwise distance $10$, and let 
	$Y$ be a maximal subset of vertical paths of $W''$ of pairwise distance $10$.
	Whenever a path $P\in X$ and a path $Q\in Y$ intersect, we reroute the two paths locally around their 
	intersection to avoid red vertices of degree $3$ as follows.

	Let $u$, $v$, $w$, $x$ be consecutive degree-$3$-vertices
	on $Q$ with  $u,v$ on $P$.
	Assume $u$, $v$, $w$, $x$ appear in this order when walking along $Q$ from top to bottom, and w\,l.\,o.\,g.\ assume $u$ is red (otherwise walk along $Q$ from bottom to top).
	Since $H'$ is white, $w$ is green $x$ is red.
	Now we reroute $P$ and $Q$ locally, such that after rerouting, both degree-3-vertices at the intersection of $P$ and $Q$ 
are green. The rerouting is 
	shown in Figure~\ref{fig:case-a}.
	Note that there is enough space around the intersection, because we 
	only use paths in $X \cup Y$.

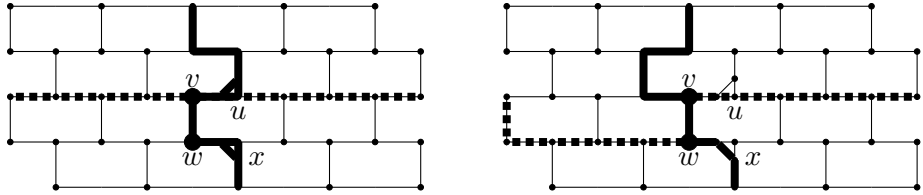
\begin{figure}[h]
\centering
	\newcommand{\Num}{3} 
	\newcommand{\Mum}{8} 

\begin{tikzpicture}[scale=.6, v/.style= {circle,draw,fill, scale=.2pt}, ]
                        foo/.style= {draw,circle,inner sep=2pt,fill}]

\foreach \j in {0,...,9}{%
         \foreach \i in {0,...,2}{
				\node[v] (h\j;\i) at ({\j},{\i}) {}; 
			}
}
\foreach \j in {0,2,...,\Mum}{%
				\node[v] (h\j;\Num) at ({\j},{\Num}) {}; 
				\node[v] (hh\j;-1) at ({\j+1},{-1}) {}; 
}

\foreach \j in {0,...,\Mum}{%
         \foreach \i in {0,...,2}{
				\draw ({\j},{\i}) -- ({\j+1},{\i}) ; 
			}
}
\foreach \j in {0,...,7}{%
				\draw ({\j},{\Num}) -- ({\j+1},{\Num}) ; 
}

	\foreach \j in {0,...,7}{%
				\draw ({\j+2},{-1}) -- ({\j+1},{-1}) ; 
}

\foreach \j in {0,2,...,\Mum}{%
			\foreach \i in {0,2,...,\Num}{
				\draw ({\j},{\i}) -- ({\j},{\i+1}) ;  
				\draw ({\j+1},{\i}) -- ({\j+1},{\i-1}) ; 
			}
    }
\foreach \j in {0,...,\Mum}{%
				\draw[line width=3pt, dash pattern={on 0 off 2 on 3 off 2 on 3 off 2 on 3 off 20pt}] ({\j},{1}) -- ({\j+1},{1}) ; 
	}

	 \draw[line width=3pt] (4,3) -- (4,2) ; 
	 \draw[line width=3pt] (4,1) -- (4,0) ; 
	 \draw[line width=3pt] (5,2) -- (5,1) ; 
	 \draw[line width=3pt] (5,0) -- (5,-1) ; 
	 \draw[line width=3pt] (4,0) -- (5,0) ; 
	 \draw[line width=3pt] (4,1) -- (5,1) ; 
	 \draw[line width=3pt] (4,2) -- (5,2) ; 

\node[v] (rot1) at (4.6,0) {}; 
\node[v] (rot2) at (5,-0.4) {}; 
	\draw[line width=3pt] (rot1) -- (rot2) ; 
\node[v] (rot1) at (4.6,1) {}; 
\node[v] (rot2) at (5,1.4) {}; 
	\draw[line width=3pt] (rot1) -- (rot2) ;

\node[v,scale=3pt] (gruen1) at (4,0) {}; 
\node[v,scale=3pt] (gruen2) at (4,1) {}; 

\draw (5,1) {} node[below]{$u$}; 
\draw (4,1) {} node[above]{$v$}; 
\draw (4,0) {} node[below]{$w$}; 
\draw (5,0) {} node[below right]{$x$}; 
\end{tikzpicture}
$\quad\quad$
\begin{tikzpicture}[scale=.6, v/.style= {circle,draw,fill, scale=.2pt}, ]
                        foo/.style= {draw,circle,inner sep=2pt,fill}]

\foreach \j in {0,...,9}{%
         \foreach \i in {0,...,2}{
				\node[v] (h\j;\i) at ({\j},{\i}) {}; 
			}
}
\foreach \j in {0,2,...,\Mum}{%
				\node[v] (h\j;\Num) at ({\j},{\Num}) {}; 
				\node[v] (hh\j;-1) at ({\j+1},{-1}) {}; 
}

\foreach \j in {0,...,\Mum}{%
         \foreach \i in {0,...,2}{
				\draw ({\j},{\i}) -- ({\j+1},{\i}) ; 
			}
}
\foreach \j in {0,...,7}{%
				\draw ({\j},{\Num}) -- ({\j+1},{\Num}) ; 
}

	\foreach \j in {0,...,7}{%
				\draw ({\j+2},{-1}) -- ({\j+1},{-1}) ; 
}

\foreach \j in {0,2,...,\Mum}{%
			\foreach \i in {0,2,...,\Num}{
				\draw ({\j},{\i}) -- ({\j},{\i+1}) ;  
				\draw ({\j+1},{\i}) -- ({\j+1},{\i-1}) ; 
			}
    }
\foreach \j in {4,...,\Mum}{%
				\draw[line width=3pt, dash pattern={on 0 off 2 on 3 off 2 on 3 off 2 on 3 off 20pt}] ({\j},{1}) -- ({\j+1},{1}) ; 
	}
\foreach \j in {0,...,3}{%
				\draw[line width=3pt, dash pattern={on 0 off 2 on 3 off 2 on 3 off 2 on 3 off 20pt}] ({\j},{0}) -- ({\j+1},{0}) ; 
	}
	
\draw[line width=3pt,  dash pattern={on 0 off 2 on 3 off 2 on 3 off 2 on 3 off 20pt}] (0,0) -- (0,1) ; 

	 \draw[line width=3pt] (3,1) -- (4,1) ; 
	 \draw[line width=3pt] (4,3) -- (4,2) ; 
	 \draw[line width=3pt] (4,1) -- (4,0) ; 
	 \draw[line width=3pt] (3,2) -- (3,1) ; 
	 \draw[line width=3pt] (5,-0.4) -- (5,-1) ; 
	 \draw[line width=3pt] (4,0) -- (4.6,0) ; 
	 \draw[line width=3pt] (4,2) -- (3,2) ; 
\node[v] (rot1) at (4.6,0) {}; 
\node[v] (rot2) at (5,-0.4) {}; 
	\draw[line width=3pt] (rot1) -- (rot2) ; 
\node[v] (rot1) at (4.6,1) {}; 
\node[v] (rot2) at (5,1.4) {}; 
	\draw (rot1) -- (rot2) ; 
	
\node[v,scale=3pt] (gruen1) at (4,0) {}; 
\node[v,scale=3pt] (gruen2) at (4,1) {}; 

\draw (5,1) {} node[below]{$u$}; 
\draw (4,1) {} node[above]{$v$}; 
\draw (4,0) {} node[below]{$w$}; 
\draw (5,0) {} node[below right]{$x$}; 
	 
\end{tikzpicture}
\caption{Rerouting in the proof of Lemma~\ref{lem:tidy-stonewall}.}
\label{fig:case-a}
\end{figure}

	Rerouting in this manner for every pair of paths in $X$ and $Y$, 
	we end up with a large subwall of $W'$ that is green, 
	which is also an induced homogeneous stone wall in $W$.
\end{proof}

\begin{proof}[Proof of Theorem~\ref{thm:ind-wall}]
Let $H$ be a graph and let $G$ be a graph
	excluding $H$ as a minor. We may assume $G$ is connected. Let $c_H$ be as in 
	Corollary~\ref{cor:contractions} and let $h\in \mathbb N$ be sufficiently large, and let
	$k=8h$.
	Assume $\tw(G)\geq c_H\cdot (k+1)^2$. Then $G$ contains an induced subgraph $G'$, 
	such that $G'$ contains $\trigrid{k}$ as a contraction, witnessed by a contraction mapping $\phi\colon G'\to \trigrid{k}$. The graph $\trigrid{k}$ contains
	$(2h)^2$ graphs $\trigrid{4}$. We pick every second row of graphs $\trigrid{4}$, and every second
	graph $\trigrid{4}$ of the row allows us to find an induced fork or an induced semi-fork in $G'$ as follows. Assume the vertices of $\trigrid{4}$ are $(1,1),\ldots, (4,4)$.
	Let $A':=\phi^{-1}((1,1)), A:=\phi^{-1}((2,1)), B':=\phi^{-1}((1,4)), B:=\phi^{-1}((1,3)), C':=\phi^{-1}((4,1)), C:=\phi^{-1}((3,2)),$ and $S:=\phi^{-1}((2,2))$.
	 Lemma~\ref{lem:forks} yields a fork or a semi-fork in 
	$G'[A\cup A'\cup B\cup B'\cup C\cup C'\cup S]$ and hence in $G'$ (cf.\ Figure~\ref{fig:fork-in-grid}).

	\begin{figure}
\centering
{
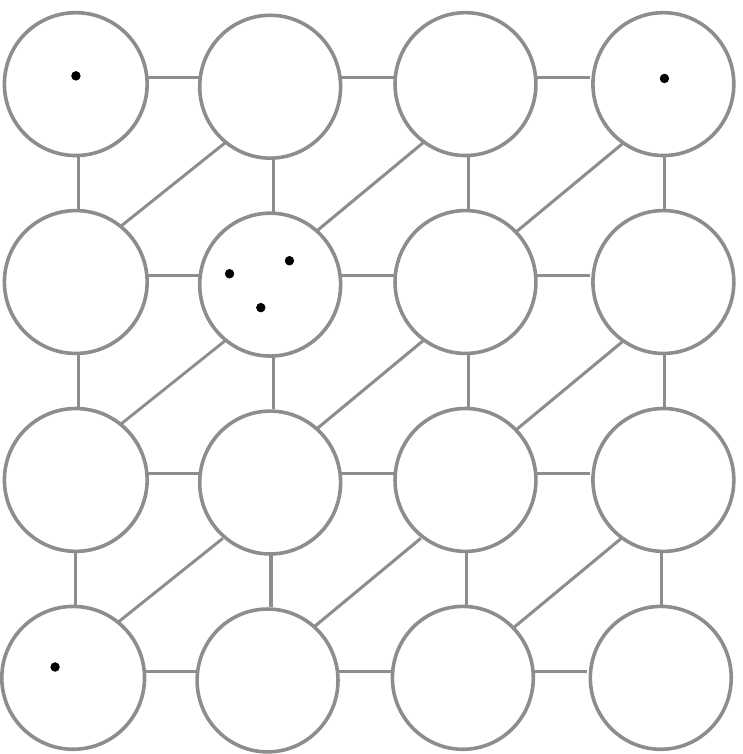
}

	\caption{Proof of Theorem~\ref{thm:ind-wall}:  Using Lemma~\ref{lem:forks} to find an induced fork or semi-fork in $G'$.}
	\label{fig:fork-in-grid}
\end{figure}

	These forks can be combined into a large stone wall by adding induced paths to connect the forks or semi-forks
	appropriately, cf.\ Figure~\ref{fig:fromG44toStonewall}. 
	\begin{figure}[h]
\centering
   \def\svgwidth{10.5cm}
		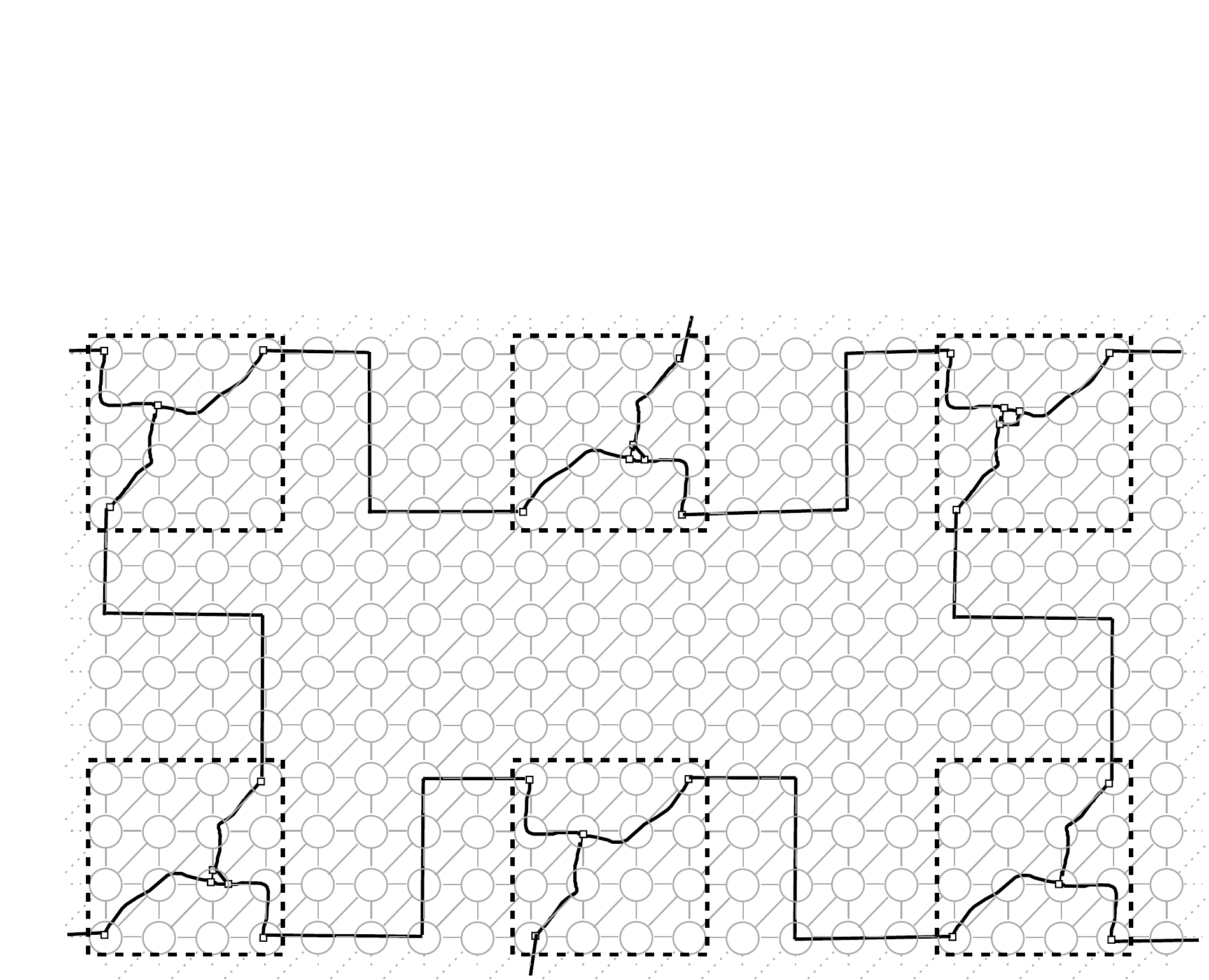

		\caption{Proof of Theorem~\ref{thm:ind-wall}: Finding a stone wall in $\trigrid{k}$.}
	\label{fig:fromG44toStonewall}
\end{figure}

  Hence $G$ contains a large stone wall as induced subgraph, and we can use Lemma~\ref{lem:tidy-stonewall} to complete the proof.
 \end{proof}

Let us remark that the function $f_H$ in Theorem~\ref{thm:ind-wall} is computable.

\begin{corollary}\label{cor:odd-signable-ehf}
	For every fixed graph $H$, the class of (theta, prism)-free
    graphs that do not contain $H$ as a minor has bounded
    tree-width. In particular, even-hole-free graphs
    that do not contain $H$ as a minor have bounded tree-width.
\end{corollary}
\begin{proof}
	A large wall contains a theta, and the line graph 
	of a large wall contains a prism.
\end{proof}

The following corollary reproves a theorem from~\cite{SilvaSS10}.

\begin{corollary}
	Planar even-hole-free graphs have bounded 
	tree-width.
\end{corollary}

\begin{proof}
	This follows from Corollary~\ref{cor:odd-signable-ehf} because
	planar graphs exclude $K_5$ as a minor.
\end{proof}


%% file: svg-inkscape/2fork-in-grid_svg-tex.pdf_tex
\begingroup%
  \makeatletter%
  \providecommand\color[2][]{%
    \errmessage{(Inkscape) Color is used for the text in Inkscape, but the package 'color.sty' is not loaded}%
    \renewcommand\color[2][]{}%
  }%
  \providecommand\transparent[1]{%
    \errmessage{(Inkscape) Transparency is used (non-zero) for the text in Inkscape, but the package 'transparent.sty' is not loaded}%
    \renewcommand\transparent[1]{}%
  }%
  \providecommand\rotatebox[2]{#2}%
  \newcommand*\fsize{\dimexpr\f@size pt\relax}%
  \newcommand*\lineheight[1]{\fontsize{\fsize}{#1\fsize}\selectfont}%
  \ifx\svgwidth\undefined%
    \setlength{\unitlength}{213.82003041bp}%
    \ifx\svgscale\undefined%
      \relax%
    \else%
      \setlength{\unitlength}{\unitlength * \real{\svgscale}}%
    \fi%
  \else%
    \setlength{\unitlength}{\svgwidth}%
  \fi%
  \global\let\svgwidth\undefined%
  \global\let\svgscale\undefined%
  \makeatother%
  \begin{picture}(1,1.0151533)%
    \lineheight{1}%
    \setlength\tabcolsep{0pt}%
    \put(0,0){\includegraphics[width=\unitlength,page=1]{2fork-in-grid_svg-tex.pdf}}%
    \put(0.10332164,0.92720017){\color[rgb]{0,0,0}\makebox(0,0)[lt]{\lineheight{1.87500012}\smash{\begin{tabular}[t]{l}$a$\end{tabular}}}}%
    \put(0.89663038,0.9269015){\color[rgb]{0,0,0}\makebox(0,0)[lt]{\lineheight{1.87500024}\smash{\begin{tabular}[t]{l}$b$\end{tabular}}}}%
    \put(0.40595587,0.65872555){\color[rgb]{0,0,0}\makebox(0,0)[lt]{\lineheight{1.87500012}\smash{\begin{tabular}[t]{l}$u$\end{tabular}}}}%
    \put(0.04742641,0.07785288){\color[rgb]{0,0,0}\makebox(0,0)[lt]{\lineheight{1.87500012}\smash{\begin{tabular}[t]{l}$c$\end{tabular}}}}%
    \put(0.36386436,0.58857323){\color[rgb]{0,0,0}\makebox(0,0)[lt]{\lineheight{1.87500012}\smash{\begin{tabular}[t]{l}$v$\end{tabular}}}}%
    \put(0.30072709,0.66574111){\color[rgb]{0,0,0}\makebox(0,0)[lt]{\lineheight{1.87500012}\smash{\begin{tabular}[t]{l}$w$\end{tabular}}}}%
    \put(0,0){\includegraphics[width=\unitlength,page=2]{2fork-in-grid_svg-tex.pdf}}%
    \put(0.02011703,0.51842092){\color[rgb]{0,0,0}\makebox(0,0)[lt]{\lineheight{1.87500012}\smash{\begin{tabular}[t]{l}$A$\end{tabular}}}}%
    \put(0.02011703,0.98844221){\color[rgb]{0,0,0}\makebox(0,0)[lt]{\lineheight{1.87500012}\smash{\begin{tabular}[t]{l}$A'$\end{tabular}}}}%
    \put(0.29371184,0.72186272){\color[rgb]{0,0,0}\makebox(0,0)[lt]{\lineheight{1.87500012}\smash{\begin{tabular}[t]{l}$S$\end{tabular}}}}%
    \put(0.29371184,0.46229971){\color[rgb]{0,0,0}\makebox(0,0)[lt]{\lineheight{1.87500012}\smash{\begin{tabular}[t]{l}$C$\end{tabular}}}}%
    \put(0.93911504,0.99545746){\color[rgb]{0,0,0}\makebox(0,0)[lt]{\lineheight{1.87500012}\smash{\begin{tabular}[t]{l}$B'$\end{tabular}}}}%
    \put(0.020117,0.18870881){\color[rgb]{0,0,0}\makebox(0,0)[lt]{\lineheight{1.87500012}\smash{\begin{tabular}[t]{l}$C'$\end{tabular}}}}%
    \put(0.66552021,0.99545746){\color[rgb]{0,0,0}\makebox(0,0)[lt]{\lineheight{1.87500012}\smash{\begin{tabular}[t]{l}$B$\end{tabular}}}}%
    \put(0.52521516,0.73589343){\color[rgb]{0,0,0}\makebox(0,0)[lt]{\lineheight{1.87500012}\smash{\begin{tabular}[t]{l}$P$\end{tabular}}}}%
    \put(0.04817801,0.75693939){\color[rgb]{0,0,0}\makebox(0,0)[lt]{\lineheight{1.87500012}\smash{\begin{tabular}[t]{l}$Q$\end{tabular}}}}%
  \end{picture}%
\endgroup%

%% file: svg-inkscape/fromG44toStonewall_svg-tex.pdf_tex
\begingroup%
  \makeatletter%
  \providecommand\color[2][]{%
    \errmessage{(Inkscape) Color is used for the text in Inkscape, but the package 'color.sty' is not loaded}%
    \renewcommand\color[2][]{}%
  }%
  \providecommand\transparent[1]{%
    \errmessage{(Inkscape) Transparency is used (non-zero) for the text in Inkscape, but the package 'transparent.sty' is not loaded}%
    \renewcommand\transparent[1]{}%
  }%
  \providecommand\rotatebox[2]{#2}%
  \newcommand*\fsize{\dimexpr\f@size pt\relax}%
  \newcommand*\lineheight[1]{\fontsize{\fsize}{#1\fsize}\selectfont}%
  \ifx\svgwidth\undefined%
    \setlength{\unitlength}{547.17745bp}%
    \ifx\svgscale\undefined%
      \relax%
    \else%
      \setlength{\unitlength}{\unitlength * \real{\svgscale}}%
    \fi%
  \else%
    \setlength{\unitlength}{\svgwidth}%
  \fi%
  \global\let\svgwidth\undefined%
  \global\let\svgscale\undefined%
  \makeatother%
  \begin{picture}(1,0.81003966)%
    \lineheight{1}%
    \setlength\tabcolsep{0pt}%
    \put(0,0){\includegraphics[width=\unitlength,page=1]{fromG44toStonewall_svg-tex.pdf}}%
  \end{picture}%
\endgroup%

%% file: Subcubic.tex

In this Section, we prove that even-hole subcubic graphs can be
described by a structure theorem, that implies tree-width at
most~3. In fact our result is for a more general class: (theta,
prism)-free subcubic graphs.

A wheel that is not a pyramid is a \emph{proper wheel}.  
A {\em sector} of a wheel $(H,x)$ is a subpath of
$H$ whose endnodes are adjacent to~$x$, and whose internal vertices
are not.

An {\em extended prism} is a graph made of five vertex-disjoint chordless
paths of length at least 1 $A = a \dots x$, $A' = x \dots a'$,
$B = b \dots y$, $B' = y \dots b'$, $C = c \dots c'$ such that $abc$
is a triangle, $a'b'c'$ is a triangle, $xy$ is an edge and no edges
exist between the paths except $xy$ and those of the two triangles
(see Figure~\ref{fig:extended-prism}).

  \begin{figure}[ht]
  	\centering \includegraphics[scale=0.6]{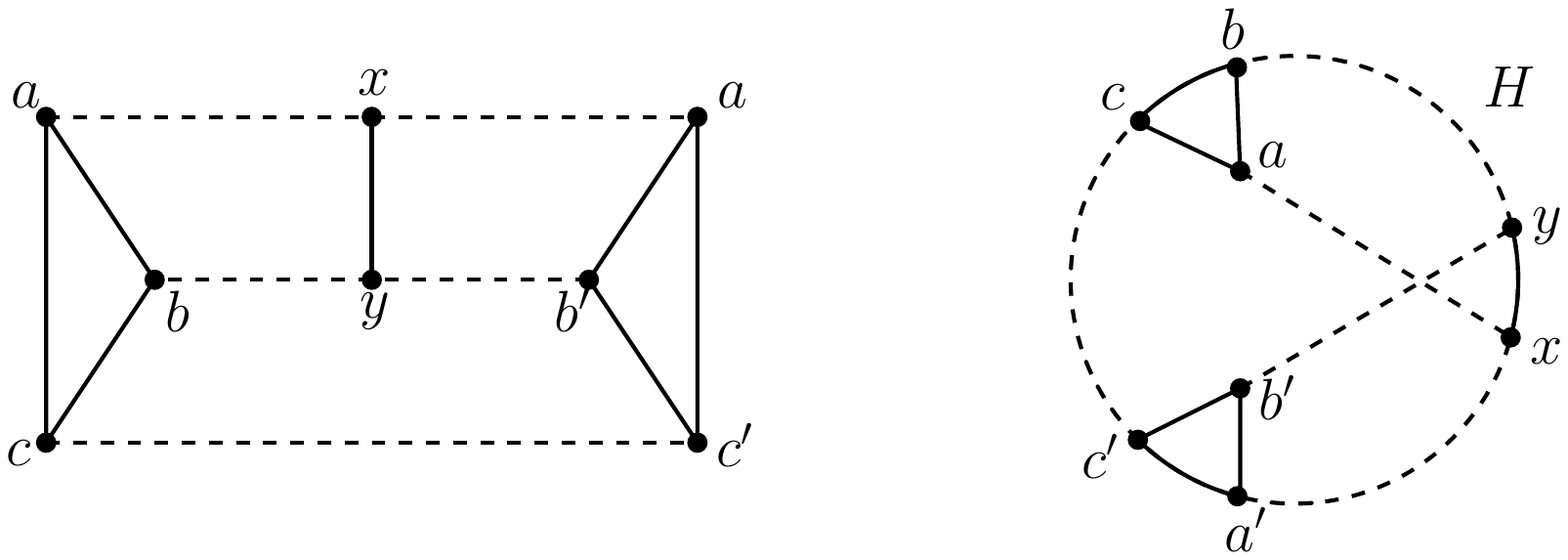}
  	\caption{Two different drawings of an extended prism}
  	\label{fig:extended-prism}
  \end{figure}

A subset (possibly empty) of
vertices $S \subseteq V(G)$ is a {\em separator} of~$G$ if $G \sm S$
contains at least two connected components.  A {\em clique separator}
is a separator $S$ that is a clique. 
 
  A {\em proper separation} in a graph $G$ is a triple
  $(\{a,b\}, X, Y)$ satisfying the following.
  \begin{enumerate}[label = (\roman*)]
    \item $\{a,b\}$, $X$, $Y$ are disjoint, non-empty and $V(G)
      =\{a,b\} \cup X \cup Y$. 
    \item There are no edges from $X$ to~$Y$.
      \item $a$ and $b$ are non-adjacent.
  \item $a$ and $b$ have exactly two neighbors in~$X$.
  \item $a$ and $b$ have exactly one neighbor in~$Y$.
  \item There exists a path from $a$ to $b$ with interior in~$X$, and there
	  exists a path from $a$ to $b$ with interior in~$Y$.
  \item $G[Y\cup \{a, b\}]$ is not a
	  chordless path from $a$ to~$b$.
\end{enumerate}

A {\em proper separator} of~$G$ is a pair $\{a, b\}\subseteq V(G)$
such that there exists a proper separation $(\{a,b\}, X, Y)$.

Let $\mathcal{C}$ be the class of (theta, prism)-free subcubic graphs.
The \emph{cube} is the graph made of a hole $v_1v_2\dots v_6v_1$ and
two non-adjacent vertices $x$ and $y$ such that
$N_H(x) = \{v_1, v_3, v_5\}$ and $N_H(y) = \{v_2, v_4, v_6\}$.  Call a
graph in~$\mathcal{C}$ {\em basic} if it is isomorphic to a chordless
cycle, a clique of size at most 4, the cube, a proper wheel, a
pyramid, or an extended prism.  An example of graph in $\mathcal C$
that is not basic is provided in Figure~\ref{fig:ruggball}.

\begin{figure}[ht]
  \centering \includegraphics[scale=0.6]{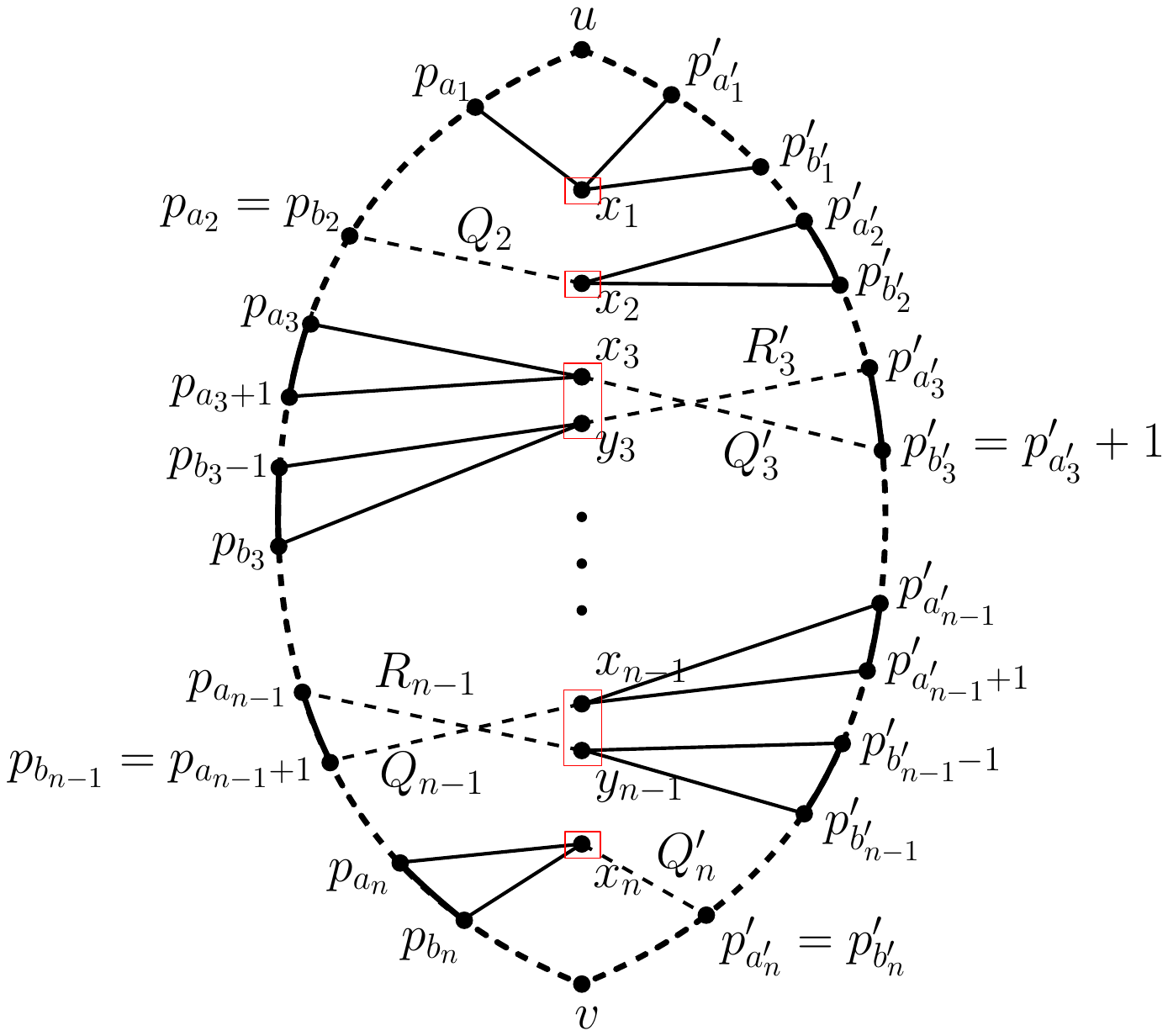}
  \caption{An example of non-basic graph in~$\cal C$}
  \label{fig:ruggball}
\end{figure}

We need the following lemma.

\begin{lemma}
\label{lem:hole-attachment}
Let $G$ be a theta-free subcubic graph, let $H$ be a hole in~$G$, and
$v \in G \sm H$. Then $v$ has at most three neighbors in~$H$, and if
$v$ has exactly two neighbors in $H$, then they are adjacent.
\end{lemma}

\begin{proof}
  Let $v \in G \sm H$.  Since $G$ is subcubic, $d_H(v) \leq 3$.  If
  $v$ has exactly two neighbors in~$H$, but they are non-adjacent then
  $G[H \cup \{v\}]$ would induce a theta, a contradiction.
\end{proof}

The main theorem of this section is the following. 

\begin{theorem}
\label{th:main-thm-max3}
Let $G$ be a (theta, prism)-free subcubic graph. Then one of the
following holds:
\begin{itemize}
\item $G$ is a basic graph;
\item $G$ has a clique separator of size at most~2;
\item $G$ has a proper separator.
\end{itemize}
\end{theorem}

\begin{proof}	
  Let $G$ be a (theta, prism)-free subcubic graph. We may assume that
  $G$ has no clique separator (and is in particular connected for
  otherwise the empty set is a clique separator).
	
  \begin{claim}
    \label{c:cubek4}
    We may assume that $G$ is ($K_4$, cube)-free.
  \end{claim}

  \bpc If $G$ contains $K_4$, then since $G$ is a subcubic connected
  graph, $G = K_4$, so $G$ is basic.  The proof is similar when $G$
  contains the cube.  \epc

  \begin{claim}
    \label{cl:wheel}
    We may assume that $G$ does not contain a proper wheel.
  \end{claim}
	
  \bpc Let $W=(H, x)$ be a proper wheel in~$G$. Let $a, b, c$, be the
  three neighbors of~$x$.  We call
  $A$ ($B$, $C$, resp.) the path of~$H$ from $b$ to~$c$ (from
  $a$ to~$c$, from $a$ to~$b$, resp.) that does not contain~$a$ ($b$,
  $c$, resp.).

  Suppose that some vertex $y$ of~$G\sm W$ has neighbors in the three
  sectors of~$W$, say $a'$ in~$A$, $b'$ in~$B$, and $c'$
  in~$C$. Hence, $a$, $c'$, $b$, $a'$, $c$, and $b'$ appear in this
	order along $H$. If $ac' \notin E(G)$, then 
	$xaBb'$, $xcBb'$, and
  $xbCc'yb'$ induce a theta, so $ac' \in E(G)$. Symmetrically, $c'b$,
  $ba'$, $a'c$, $b'c$, and $b'a$ are all in~$E(G)$, so $H$, $x$, and
  $y$ induce a cube, a contradiction to~(\ref{c:cubek4}). It follows
  that every vertex has neighbors in at most two sectors of~$W$.

  If $G=W$, then $G$ is basic, so let $L$ be a component of $G\sm
  W$. Note that $N(L)$ contains at least two vertices since $G$ has no
  clique separator. If $N(L)$ is included in a sector of~$W$, then the
  ends of this sector form a proper separator.  We may therefore
  assume that $N(L)$ intersects at least two sectors of~$W$.

  Since $L$ is connected, it contains a path $P = u\dots v$ such that
  $u$ has neighbors in a sector of~$W$ (say $C$ up to symmetry), and
  $v$ has neighbors in another sector of~$W$ (say $A$ up to
  symmetry). Suppose that $P$ is minimal with respect to this
  property. Then either $u=v$ and by the
  second paragraph of this proof, $u$ has no neighbor in~$B$; or
  $u \neq v$ and, by minimality of $P$, $u$ has neighbors only in $C$, 
	$v$ has neighbors only
	in~$A$, and the interior of $P$ is anticomplete to~$W$. In each
  case, we let $u'$ be the neighbor of~$u$ in~$C$ closest to~$a$ along
  $C$ and we let $v'$ be the neighbor of~$v$ in~$A$ closest to~$c$
  along $A$. Note that because~$u'$ and $v'$ exist, $ab\notin E(G)$ and
  $bc\notin E(G)$. So, $ac\notin E(G)$ for otherwise, $(W, x)$ would
  form a pyramid and be a non-proper wheel. Now, the three paths
  $axc$, $B$, and $aCu'uPvv'Ac$ form a theta, a contradiction.  \epc
  \medskip

  \begin{claim}
    \label{cl:ext-prism}
    We may assume that $G$ does not contain an extended prism.
  \end{claim}
	
  \bpc Let $W$ be an extended prism in~$G$, with notation as in the
  definition.  Suppose that some vertex $z$ of~$G\sm W$ has neighbors
  in three distinct paths among $A, A', B, B'$, and $C$, and call $Q, R, S$
  these three paths (so $\{Q, R, S\} \subseteq \{A, A', B, B', C\}$).
  It is easy to check that some hole $H$ of~$W$ contains $Q$ and
  $R$. By Lemma~\ref{lem:hole-attachment}, $z$ must have three
  neighbors in~$H$, so $H$ and $z$ form a proper wheel, a
  contradiction to~(\ref{cl:wheel}).
        
  If $G=W$, then $G$ is basic, so let $L$ be a component of $G\sm
  W$. Note that $N(L)$ contains at least two vertices since $G$ has no
  clique separator.  If $N(L)$ is included in one of~$V(A)$, $V(A')$,
  $V(B)$, $V(B')$, or $V(C)$, then the ends of this path form a proper
  separator. We may therefore assume that $N(L)$ intersects at least
  two paths in~$\{A, A', B, B', C\}$.

  Since $L$ is connected, it contains a path $P = u\dots v$ such that
  $u$ has neighbors in a path $Q\in \{A, A', B, B', C\}$ and $v$ has
  neighbors in another path $R\in \{A, A', B, B', C\}$. Suppose that
  $P$ is minimal with respect to this property. So by the minimality
  of~$P$, either $u=v$ and by the second paragraph of this proof,
  $u=v$ has no neighbor in $\{A, A', B, B', C\}\sm \{Q, R\}$; or
  $u \neq v$ and $u$ has neighbors only in~$Q$, $v$ has neighbor only
  in~$R$ and the interior of~$P$ is anticomplete to~$W$.

  Note that each of~$N_{Q}(u)$ and $N_{R}(v)$ is a vertex or an edge.
  For otherwise, suppose that $u$ has two non-adjacent neighbors
  in~$Q$ (resp.\ in $R$).  Since $G$ is subcubic and $Q$ (resp.\ $R$)
  can be completed to a hole $J$ of~$W$, by
  Lemma~\ref{lem:hole-attachment}, $u$ has three pairwise non-adjacent
  neighbors in~$J$, so $G$ contains a proper wheel, a contradiction
  to~(\ref{cl:wheel}).  We may now break into four cases.

  \medskip
  \noindent{Case 1:} $\{Q, R\} = \{A, A'\}$ or $\{Q, R\} = \{B, B'\}$.
  Up to symmetry, we suppose $Q=A$ and $R=A'$.  Then, $P$ can be used
  to find a path from $a$ to~$a'$ that does not contain~$x$, and that
  together with $B$, $B'$ and $C$ form a prism, a contradiction.
        
  \medskip
  \noindent{Case 2:} $\{Q, R\} = \{A, B\}$ or $\{Q, R\} = \{A', B'\}$.
  Up to symmetry, we suppose $Q=A$ and $R=B$.  If $u$ has two adjacent
  neighbors in~$A$, then $A$, $A'$, $C$, a subpath of~$B$, and $P$
  form a prism. So, $u$ has exactly one neighbor in~$A$, and
  symmetrically, $v$ has exactly one neighbor in~$B$.  So, $A$, $B$,
  and $P$ form a theta.
        
  \medskip
  \noindent{Case 3:} $\{Q, R\} = \{A, B'\}$ or $\{Q, R\} = \{B, A'\}$.
  Up to symmetry, we suppose $Q=A$ and $R=B'$.  If $u$ has two
  adjacent neighbors in~$A$, then $A$, $A'$, $C$, a subpath of~$B'$,
  and $P$ form a prism. So, $u$ has exactly one neighbor in~$A$, and
  symmetrically, $v$ has exactly one neighbor in~$B'$.  So, $A$, $B'$,
  $C$, and $P$ form a theta.

  \medskip
  \noindent{Case 4:} $\{Q, R\}$ is one of~$\{A, C\}$, $\{A', C\}$,
  $\{B, C\}$ or $\{B', C\}$.  Up to symmetry, we suppose $Q=A$ and
  $R=C$.  If $v$ has two adjacent neighbors in~$C$, then $C$, $B$,
  $B'$, a subpath of~$A$ and $P$ form a prism. So, $v$ has exactly one
  neighbor in~$C$. So, $C$, $B$, $A'$, a subpath of~$A$, and $P$ form
  a theta.  \medskip
        
  \epc \medskip

  \begin{claim}
    \label{cl:pyramid}
    We may assume that $G$ does not contain a pyramid.
  \end{claim}
	
  \bpc Let $W$ be a pyramid with notation as in the definition.  First
  note that a vertex $v\in V(G\sm W)$ cannot have neighbors in the
  three paths $P_1, P_2$, and $P_3$, for otherwise there exists a
  theta from $v$ to~$x$.

  If $G=W$, then $G$ is basic, so let $L$ be a component of $G\sm
  W$. Note that $N(L)$ contains at least two vertices since $G$ has no
  clique separator.  If $N(L)$ is included in one of~$P_1$, $P_2$, or
  $P_3$, then the ends of this path form a proper separator, or a
  clique separator when this path has length~$1$.  We may therefore
  assume that $N(L)$ intersects at least two paths.

  So, since $L$ is connected, it contains a path $P = u\dots v$ such
  that $u$ has neighbors in a path $P_i$ (say $P_1$ up to symmetry),
  and $v$ has neighbors in another path $P_j$ (say $P_2$ up to
  symmetry). Suppose that $P$ is minimal with respect to this
  property. So by minimality, either $u=v$ and by the first paragraph
  of this proof, $u=v$ has no neighbor in $P_3$; or $u \neq v$ and $u$
  has neighbors only in~$P_1$, $v$ has neighbor only in~$P_2$, and
  the interior of $P$ is anticomplete to~$W$.

  Note that each of~$N_{P_1}(u)$ and $N_{P_2}(v)$ is a vertex or an
  edge.  If $u=v$, this is because $G$ contains no proper wheel
  by~(\ref{cl:wheel}).  If $u\neq v$, this is because $u$ and $v$ have
  degree at most~3 and we apply Lemma~\ref{lem:hole-attachment}.

  If $N_{P_1}(u)$ and $N_{P_2}(v)$ are both edges, then $u\neq v$
  (because $G$ is subcubic), so $P_1$, $P_2$, and $P$ form a prism.
  If each of~$N_{P_1}(u)$ and $N_{P_2}(v)$ is a vertex, then $P_1$,
  $P_2$, and $P$ form a theta. So, up to symmetry, $N_{P_1}(u)$ is a
  vertex $u'$, $N_{P_2}(v)$ is an edge $yz$ (where $x, y, z, a_2$
  appear in this order along $P_2$).  If $u'\neq x_1$, then
  $V(P)\cup V(W) \sm V(zP_2a_2)$ induces a theta from $u'$ to~$x$, so
  $u'=x_1$.  Hence, $W$ and $P$ form an extended prism, a
  contradiction to~(\ref{cl:ext-prism}) \epc

  \begin{claim}
    \label{cl:hole}
    We may assume that $G$ does not contain a hole.
  \end{claim}
	
  \bpc Let $W$ be a hole in~$G$. First note that a vertex
  $v\in V(G\sm W)$ cannot have three neighbors in~$W$, for otherwise
  $v$ and $W$ would form a proper wheel or a pyramid.  So, by
  Lemma~\ref{lem:hole-attachment}, every vertex of $G\sm W$ has at
  most one neighbor in~$W$, or exactly two neighbors in~$W$ that are
  adjacent.

  If $G=W$, then $G$ is basic, so suppose that $L$ is a component
  of~$G\sm W$.  If $N(L)$ is included in some edge of $W$, then $G$
  has a clique separator, so suppose that there exist $a, b\in V(W)$
  that are non-adjacent and that both have neighbors in~$L$.  Since
  $L$ is connected, there exists a path $P=u \dots v$, such that $u$
  is adjacent to~$a$ and $v$ is adjacent to~$b$.  We suppose that
  $a, b, u, v$ and $P$ are chosen subject to the minimality of~$P$.  Note that
  $u\neq v$ since a vertex in~$G\sm W$ cannot have two non-adjacent
  neighbors in~$W$.

  Suppose that some internal vertex of~$P$ has a neighbor $x$ in~$W$.
  So $x$ must be adjacent to $a$, for otherwise a subpath of~$P$ from
  $u$ to a neighbor of $x$ in $P$ contradicts the minimality of $P$.
  Similarly, $x$ is adjacent to $b$. If $a$ and $b$ have two common
  neighbors in~$H$, say $x$ and $y$ (so $W=axbya$), and $x$ and $y$
  both have neighbors in the interior of~$P$, then the vertices $x$
  and $y$ together with a subpath of~$P$ contradict the minimality
  of~$P$. Hence, $x$ is the unique vertex of~$W$ with neighbors in the
  interior of~$P$.  If $u$ and $v$ have exactly two adjacent neighbors
  in~$W$, then $W$ and $P$ form an extended prism, a contradiction
  to~(\ref{cl:ext-prism}).  If exactly one of $u$ or $v$ has exactly
  two neighbors in~$W$, then $W$ and a subpath of~$P$ form a pyramid,
  a contradiction to~(\ref{cl:pyramid}).  So, $u$ and $v$ both have a
  unique neighbor in~$W$.  Now, $P$ and $H$ form a proper wheel, a
  contradiction to~(\ref{cl:wheel}).

  So, the interior of $P$ is anticomplete to~$W$. Hence, $P$ and $W$
  form a theta, a prism or a pyramid, in every case a contradiction to
  $G\in \cal C$, or to~(\ref{cl:pyramid}).  \epc

  \begin{claim}
    \label{cl:triangle}
    We may assume that $G$ does not contain a triangle.
  \end{claim}
        
  \bpc Let $W=abc$ be a triangle in~$G$.  If $G=W$, then $G$ is basic,
  so suppose that $L$ is a component of~$G\sm W$.  If $|N(L)|\leq 2$,
  then $G$ has a clique separator of size at most 2, so suppose that
  $N(L) = \{a, b, c\}$.

  Let $P= u\dots v$ be a path in~$L$ such that $u$ is adjacent to~$a$,
  $v$ is adjacent to~$b$, and suppose $P$ is minimal.  If $u\neq v$,
  then $P$, $a$, and $b$ form a hole, a contradiction
  to~(\ref{cl:hole}), so $u=v$.  By~(\ref{c:cubek4}), $u$ is
  non-adjacent to~$c$.  Hence, a path in~$L$ from $u$ to a neighbor
  of~$c$, together with $a$, would form a hole, a contradiction
  to~(\ref{cl:hole}).  \epc

  \medskip Now, by~(\ref{cl:hole}) and~(\ref{cl:triangle}), $G$ has no
  cycle.  So, $G$ is a tree. It is therefore a complete graph on at
  most two vertices (that is basic) or it a has clique separator of
  size~1.
\end{proof}

Let us point out that Theorem~\ref{th:main-thm-max3} is a full
structural description of the class of subcubic (theta, prism)-free
graphs, in the sense that every graph in the class can be obtained
from basic graphs by repeatedly applying some operations: gluing along
a (possibly empty) clique, and an operation called proper gluing that
we describe now.

Consider two graphs $G_1$ and $G_2$.  Suppose that $G_1$ contains two
non-adjacent vertices $a_1$ and $b_1$ of degree~3, and such that a
path $P_1$ from $a_1$ to $b_1$ with internal vertices all of degree~2
exists in $G_1$.  Suppose that $G_2$ contains two non-adjacent
vertices $a_2$ and $b_2$ of degree~2, and such that a path $P_2$ from
$a_2$ to $b_2$ with internal vertices all of degree~2 exists in $G_2$.
Let $G$ be the graph obtained from the disjoint union of $G_1$ and
$G_2$ by removing the internal vertices of $P_1$ and $P_2$, by
identifying $a_1$ and $a_2$, and by identifying $b_1$ and $b_2$.  We
say that $G$ is obtained from $G_1$ and $G_2$ by a \emph{proper
  gluing}.

We omit the details of the proof and just sketch it. We apply
Theorem~\ref{th:main-thm-max3}.  If $G$ is basic, there is nothing to
prove. If $G$  has a clique
separator, it is obtained by two smaller graphs by gluing along a clique. If
$G$ has a proper separation, then it is obtained from smaller graphs
by a proper gluing.

\begin{figure}
  \begin{center}
    \includegraphics[height=3cm]{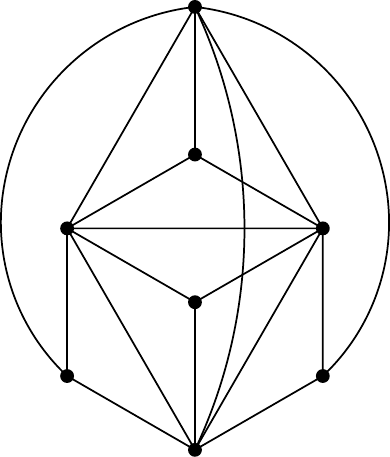}\hspace{1cm}
    \includegraphics[height=3cm]{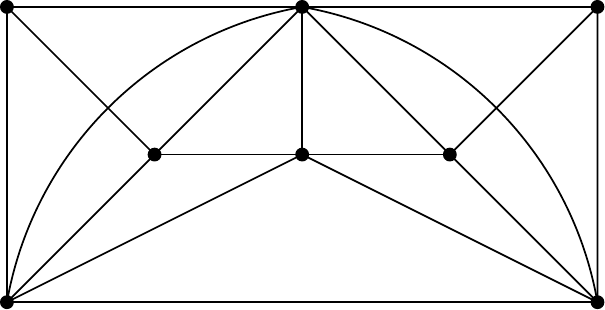}
  \end{center}
  \caption{Two chordal graphs with clique number~4\label{f:chordal}}
\end{figure}

\begin{corollary}
  Every subcubic (theta, prism)-free graph (and therefore every
  even-hole-free subcubic graph) has tree-width at most~3. 
\end{corollary}

\begin{proof}
  The proof is by induction. Let us first prove that all basic graphs
  have tree-width at most~$3$.  First observe that
  contracting an edge with one vertex of degree~$2$ preserves the
  tree-width.  It follows that all basic graphs, except the cube and the
  extended prisms, have tree-width at most the tree-width of $K_4$,
  that is~$3$. In Figure~\ref{f:chordal}, we show a chordal graph $J$
  with $\omega(J) = 4$ that contains the cube or the smallest extended
  prism as a subgraph, showing that here again the tree-width is at
  most~$3$.

  Also, it is easy to check that the two operations gluing along a
  clique and proper gluing do not increase the tree-width (this can be
  proved by observing that the operation are particular cases of what
  is called clique-sum in the theory of tree-width, or by a direct
  proof using the definition of tree-width given at the beginning of
  the paper).
\end{proof}

Note also that all graphs in~$\cal C$ can be proved to be planar by an
easy induction.

%% file: Delta4a.tex
Our goal in this section is to prove that (even hole, pyramid)-free
graphs with maximum degree at most~4 have bounded tree-width. We rely on two
known theorems that we now explain.

Let $H$ be a hole in a graph and let $u$ be a vertex not in~$H$.  We
say that $u$ is \emph{major} w.r.t.\ $H$ if $N_H(u)$ is not included
in a 3-vertex path of~$H$.  We omit ``w.r.t.\ $H$'' when $H$ is clear
from the context.  

\begin{lemma}
  \label{l:majH}
  If $G$ be an (even hole, pyramid)-free graph with maximum degree at
  most~4, $H$ is a hole of $G$ and $v$ is a vertex that is major
  w.r.t.~$G$, then $v$ has exactly three neighbors that are pairwise
  non-adjacent. 
\end{lemma}

\begin{proof}
  Since $v$ is major, it has at least two neighbors in $H$.  If $v$
  has exactly two neighbors in $H$, since $v$ is major these two
  neighbors are non-adjacent. Therefore, $H$ and $v$ form a theta, a
  contradiction.  If $G$ has exactly three neighbors in $H$, then they
  are pairwise non-adjacent because $v$ is major and $G$ has no
  pyramid.  If $v$ has~4 neighbors in $H$, then $H$ and $v$ form an
  even wheel, a contradiction.
\end{proof}

When $H$ is a hole in some graph and $u$ is a vertex not in~$H$ with
at least two neighbors in~$H$, we call \emph{$u$-sector} of~$H$ any
path of~$H$ of length at least~1, whose ends are adjacent to~$u$ and
whose internal vertices are not. Observe that $H$ is edgewise
partitioned into its $u$-sectors. 

Note that by Lemma~\ref{l:majH}, when $v$ is major w.r.t.~$H$,
$(H, v)$ is a wheel, so the notion of $v$-sector in $H$ is equivalent to 
the notion of a sector the wheel $(H,v)$.
The following appeared in~\cite{DBLP:journals/corr/abs-1912-11246}.

\begin{theorem}
  \label{th:ehfPyrFree}
  Let $G$ be a graph with no even hole and no pyramid, $H$ a hole in
  $G$ and $v$ a major vertex w.r.t.\ $H$. If $C$ is a connected
  component of~$G\sm N[v]$, then there exists a $v$-sector
  $P=x\dots y$ of~$H$ such that
  $N(C) \subseteq \{x, y\} \cup (N(v) \sm V(H))$.
\end{theorem}

A graph $G$ is a \emph{ring} if its vertex-set can be partitioned into
$k\geq 3$ sets $X_1, \dots, X_k$ such that (the subscript are taken
modulo $k$):
\begin{enumerate}
\item $X_1$, \dots, $X_k$ are cliques;
\item for all $i\in \{1, \dots, k\}$, $X_i$ is anticomplete to
  $V(G) \sm (X_{i-1} \cup X_i \cup X_{i+1})$;
\item for all $i\in \{1, \dots, k\}$, some vertex of $X_i$ is complete
  to $X_{i-1} \cup X_{i+1}$;
\item for all $i\in \{1, \dots, k\}$ and all $x, x' \in X_i$, either
  $N[x] \subseteq N[x']$ or $N[x'] \subseteq N[x]$.
\end{enumerate}

A graph $G$ is a \emph{7-hyperantihole} if its vertex-set can be
partitioned into $7$ sets $X_1, \dots, X_7$ such that (the subscript
are taken modulo $k$):
\begin{enumerate}
\item $X_1$, \dots, $X_k$ are cliques;
\item for all $i\in \{1, \dots, k\}$, $X_i$ is complete to
  $V(G) \sm (X_{i-1} \cup X_i \cup X_{i+1})$;
\item for all $i\in \{1, \dots, k\}$, $X_i$ is anticomplete to
  $X_{i-1} \cup X_{i+1}$.
\end{enumerate}

The following is a rephrasing of Theorem~1.8
in~\cite{DBLP:journals/jgt/BoncompagniPV19}.  Note that
in~\cite{DBLP:journals/jgt/BoncompagniPV19}, the definition of rings
is slightly more restricted (at least 4 sets are required). We need
rings with 3 sets for later use in inductions, and slightly extending
the notion of ring cannot turn Theorem~1.8
in~\cite{DBLP:journals/jgt/BoncompagniPV19} into a false statement.

\begin{theorem}
  \label{th:rings}
  If $G$ is (theta, prism, pyramid)-free and for every hole $H$ of
  $G$, no vertex of $G$ is major w.r.t.~$H$, then $G$ is a complete
  graph, or $G$ is a ring, or $G$ is a 7-hyperantihole, or $G$ has a
  clique separator.
\end{theorem}

\begin{lemma}
  \label{l:basicNoK6}
  A complete graph, a ring, or a 7-antihole of maximum degree at
  most~4 does not contain $K_6$ as a minor. 
\end{lemma}

\begin{proof}
  For complete graphs, this is obvious since $K_5$ is the biggest
  complete graph of maximum degree at most~4.  For 7-hyperantiholes,
  the proof is also easy because each of the cliques in the definition
  must be on a single vertex, so that $|V(G)| =7$, and a $K_6$ minor
  obviously does not exists.

  So, suppose that $G$ is a ring of maximum degree at most~4 (we use
  for $G$ the notation as in the definition of rings).

  \begin{claim}
    \label{c:no3cliques}
    For all $i\in \{1, \dots, k\}$, one of $X_{i-1}$, $X_i$ or
    $X_{i+1}$ contains only one vertex. Moreover, if $|X_i|=3$, then
    $X_{i-1}$, $X_{i+1}$ and $X_{i+2}$ all contain only one vertex.
  \end{claim}
  
  \bpc Otherwise, some vertex in $X_i$ or $X_{i+1}$ has degree~5, a
  contradiction to our assumption.  \epc

  We now prove by induction on $k$ (the number of sets in the ring)
  that $G$ does not contain $K_6$ as a minor.  If $k$ equals 3 or 4,
  then by~(\ref{c:no3cliques}), we see that $|V(G)|\leq 6$, so $G$
  does not contain $K_6$ as a minor since $G$ is not $K_6$.  If $k\geq 5$,
  then by~(\ref{c:no3cliques}), there exist two distinct sets of the
  ring $X_i$, $X_j$ that are anticomplete to each other and such that
  $X_i= \{x\}$ and $|X_j|=\{x'\}$.  By the definition of rings,
  $G\sm \{x, x'\}$ has two connected components $C$ and $C'$, and it is
  straightforward to check the two graphs $G_C$ and $G_{C'}$ obtained
  from $G[C\cup \{x, x'\}]$ and $G[C'\cup \{x, x'\}]$ respectively by
  adding an edge between $x$ and $x'$ are rings (this is the place
  where we need a ring on three sets).  Also, it is straightforward to
  check that a $K_6$ minor in $G$ yields a $K_6$ minor in one of $G_C$
  or $G_{C'}$, a contradiction to the induction hypothesis.
\end{proof}

We can now prove the main theorem of this section. 

\begin{theorem}
  \label{th:noK6}
  If a graph $G$ is (even hole, pyramid)-free with maximum degree~4,
  then $G$ contains no $K_6$ as a minor.
\end{theorem}

\begin{proof}
  Suppose that $G$ is a counter-example with a minimum number of
  vertices.  So, $G$ contains $K_6$ as a minor.  By the minimality of
  $G$ and the definition of minors, it follows that $V(G)$ can be
  partitioned into six non-empty sets $B_1$, \dots, $B_6$ such that for all
  $i, j\in \{1, \dots, 6\}$, $G[B_i]$ is connected and there is at
  least one edge between $B_i$ and $B_j$.

  \medskip
  \noindent{\bf Case 1:}  $G$ contains no hole with a major vertex.

  By Theorem~\ref{th:rings} and Lemma~\ref{l:basicNoK6}, $G$ has
  a clique separator $K$. It is straightforward to check that for one
  component $C$ of $G \sm K$, the graph $G[K \cup C]$ contains $K_6$
  as a minor, a contradiction to the minimality of $G$. 

  \medskip
  \noindent{\bf Case 2:} $G$ contains a hole $H$ and a vertex $v$
  that is major w.r.t.\ $H$.

  By Lemma~\ref{l:majH}, $v$ has exactly three neighbors in $H$ that
  are pairwise non-adjacent.  Possibly, $v$ has a neighbor $w\notin H$
  (if $v$ has degree three, we set $v=w$).  Let $a, b, c$ be the three
  neighbors of~$v$ in $H$.  

  Up to symmetry, $B_5$ and $B_6$ do not contain~$a, b, c, w$.  So,
  some connected component $C$ of~$G\sm\{v, w, a, b, c\}$ contains
  $B_5\cup B_6$.  By Theorem~\ref{th:ehfPyrFree}, there exists a
  $v$-sector, say $P = a\dots b$ up to symmetry, of~$H$ such that
  $N(C) \subseteq \{a, b, w\}$.  Note that if $v=w$, then $\{a, b\}$
  is a separator of $G$, in which case the proof is easier. So in what follow,
  a reader may assume for simplicity that $v\neq w$, though what is
  written is correct even if $v=w$.

  Let $C'$ be the union of all components $X$ of
  $G\sm\{v, w, a, b, c\}$ such that $N(X) \subseteq \{a, b, w\}$.  Let
  $D$ be $V(G) \sm (C' \cup \{a, b, w\})$.  Note that $C\subseteq C'$
  and $v, c\in D$.  Note that $B_5\cup B_6 \subseteq C'$, and since
  $\{a, b, w\}$ separates $C'$ from $D$, we may assume that
  $B_4 \subseteq C'$.

  Let $S_a$ (resp.\ $S_b)$ be the $v$-sector of $H$ from $a$
  (resp.\ $b$) to $c$. Let $G'$ be the graph obtained from
  $G[C' \cup \{a, b, c, v, w\}]$ by adding the edges $ca$ and $cb$.
  Also, the edge $cw$ is added to $G'$ if and only if $w$ has a neighbor in the
  interior of the path formed by $S_a$ and $S_b$.  

  \begin{claim}
    \label{cl:GpEHF}
    $G'$ is (even hole, pyramid)-free and has maximum degree at most~4.
  \end{claim}

  \bpc Clearly $G$ has maximum degree at most~4.  Since $G'\sm c$ is
  an induced subgraph of $G$, every even hole or pyramid of $G'$ goes
  through $c$.

  Suppose that $J$ is an even hole of $G'$. Since it goes through $c$,
  up to symmetry, we may assume that $J$ goes through $cb$. If $J$
  contains $a$, then $G$ is formed of $a$, $c$, $b$ and a path $P$ of
  even length from $a$ to $b$.  So, $P$, $S_a$ and $S_b$ form an even
  hole of $G$, unless $w\in V(J)$ and $w$ has a neighbor in the
  interior of the path induced by $S_a\cup S_b$.  But this case leads
  to a contradiction, since by the definition of $G'$, we would have
  $cw\in E(G')$, so $J$ would not be a hole of $G'$.

  So, $J$ does not go through $a$. It follows that $J$ is formed by
  $bc$ and a path $Q$ of odd length from $w$ to $b$.  Note that in
  this case, $cw\in E(G)$ and $v\notin V(J)$, so $v\neq w$.  It
  follows that $Q$ and $v$ form an even hole of $G'$, a contradiction.

  Suppose that $G'$ contains a pyramid $\Pi$.  Since $\Pi$ contains
  $c$, it does not contain $v$ because $v$ dominates $c$ , in $G'$, ie
  $N_{G'}[c]\subseteq N_{G'}[v]$ (and in a pyramid, no vertex
  dominates another vertex).  If we replace $c$ by $v$ in $\Pi$, then
  we obtain an induced subgraph of $G$ that is not a pyramid since $G$
  is pyramid-free.  This implies that $cw\notin E(G')$.  So, $c$ has
  degree~2 in $\Pi$ and $w$ has no neighbor in the interior of the
  path induced by $S_a\cup S_b$.  Hence, replacing $acb$ by $S_a$ and
  $S_b$ in $\Pi$ yields a pyramid of $G$, a contradiction. \epc

  \begin{claim}
    \label{cl:GpK6}
    $G'$ contains $K_6$ as a minor.
  \end{claim}

  \bpc Suppose that $a\in B_1$, $w\in B_2$ and $b\in B_3$.  We then
  set $B'_1 = (B_1 \sm D) \cup \{c\}$, $B'_2 = (B_2 \sm D)\cup \{v\}$
  and $B'_3 = B_3 \sm D$.  We observe that there are edges from $B'_1$
  to $B'_2$, from $B'_1$ to $B'_3$ and from $B'_2$ to $B'_3$.  Also,
  each of these sets is connected in $G'$, and together with $B_4$, $B_5$ and
  $B_6$ they form a $K_6$ minor of $G'$.  We may therefore assume that
  $B_3 \cap \{a, b, w\} = \emptyset$, so that $B_3\subseteq C'$.

  If $\{a, b, w\} \cap B_1 = \{a, b\}$, then
  $\{a, b, w\}\cap B_2 = \{w\}$.  We then set
  $B'_1= (B_1\sm D) \cup \{c\}$ and $B'_2 = (B_2\sm D) \cup \{v\}$.
  We observe that there are edges from $B'_1$ to $B'_2$.  Also, these
  sets are connected in $G'$, and together with $B_3$, $B_4$, $B_5$
  and $B_6$ they form a $K_6$ minor of $G'$. Hence, we may assume that
  $\{a, b, w\} \cap B_1 \neq  \{a, b\}$, and symmetrically $\{a, b,
  w\} \cap B_2 \neq  \{a, b\}$. 

  We may assume that $a, w\in B_1$.  We set
  $B'_1= (B_1\sm D) \cup \{c, v\}$ and $B'_2 = B_2\sm D$.  We observe
  that there are edges from $B'_1$ to $B'_2$.  Also, these sets are
  connected in $G'$, and together with $B_3$, $B_4$, $B_5$ and $B_6$
  they form a $K_6$ minor of $G'$.  
\epc

  Since $G'$ is smaller than $G$, (\ref{cl:GpEHF}) and~(\ref{cl:GpK6})
  contradict the minimality of~$G$.
\end{proof}

In the next corollary, we use the function $f_H(k)$ 
as defined in Theorem~\ref{thm:ind-wall}.

\begin{corollary}
  Every (even hole, pyramid)-free graph of maximum degree at most~4
  has tree-width less than~$f_{K_6}(3)$. 
\end{corollary}

\begin{proof}
  Suppose that $G$ has tree-width at least $f_{K_6}(3)$.
  By Theorem~\ref{th:noK6}, $G$ does not contain $K_6$ as a
  minor. 
  By Theorem~\ref{thm:ind-wall}, $G$ contains a 
  $(3 \times 3)$-wall or the line graph of a
chordless $(3 \times 3)$-wall as an induced subgraph.
This yields a contradiction because the
$(3 \times 3)$-wall contains a theta, and 
the line graph of a chordless $(3 \times 3)$-wall contains a prism
(see Figure~\ref{fig:theta-grid-in-wall-subdwall}),
a contradiction to Lemma~\ref{lem:ehf-theta-prism}.
\end{proof}

\begin{figure}[ht]
	\centering \includegraphics[scale=0.55]{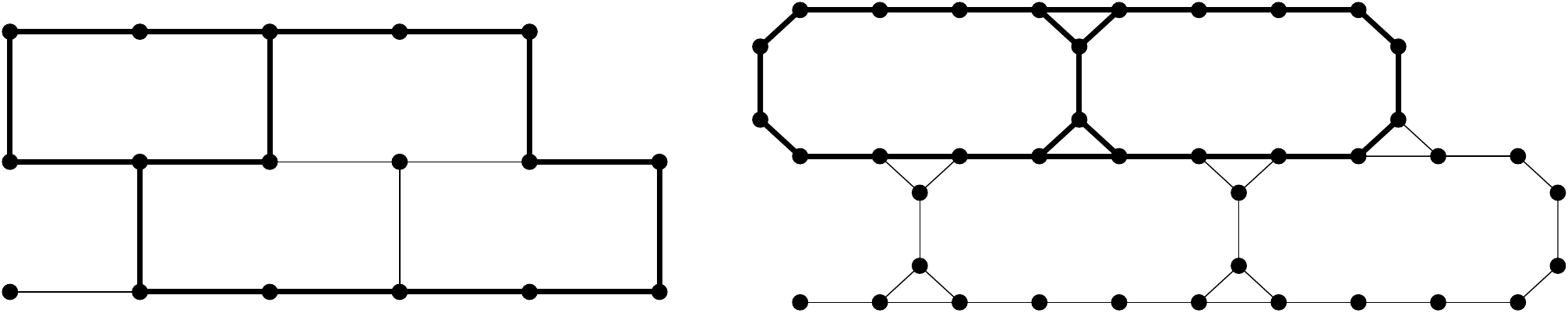}
	\caption{A theta in $(3 \times 3)$-wall and a prism in the line graph of a chordless $(3 \times 3)$-wall}
	\label{fig:theta-grid-in-wall-subdwall} 
\end{figure}

  This approach might work for maximum degree~5 with much more
technicalities, but for larger values it fails as far as we can see.

%% file: Delta4b.tex
In this section, we investigate a possible structure theorem that would describe even-hole-free graphs with maximum degree at most~4. We call {\em patterns}, the graphs that are represented on Figure~\ref{fig:wheel-family} and Figure~\ref{fig:pattern-pyramid}. Say that a graph is {\em basic} if it is a complete graph or a chordless cycle, or it can be obtained from one of the patterns, by replacing dashed lines with paths of length at least one or contracting some dashed lines into single vertices. We believe that an even-hole-free graph with maximum degree 4 must be either basic or decomposable with a clique separator or a 2-join that we define below.   

\begin{figure}[ht]
	\centering \includegraphics[scale=0.5]{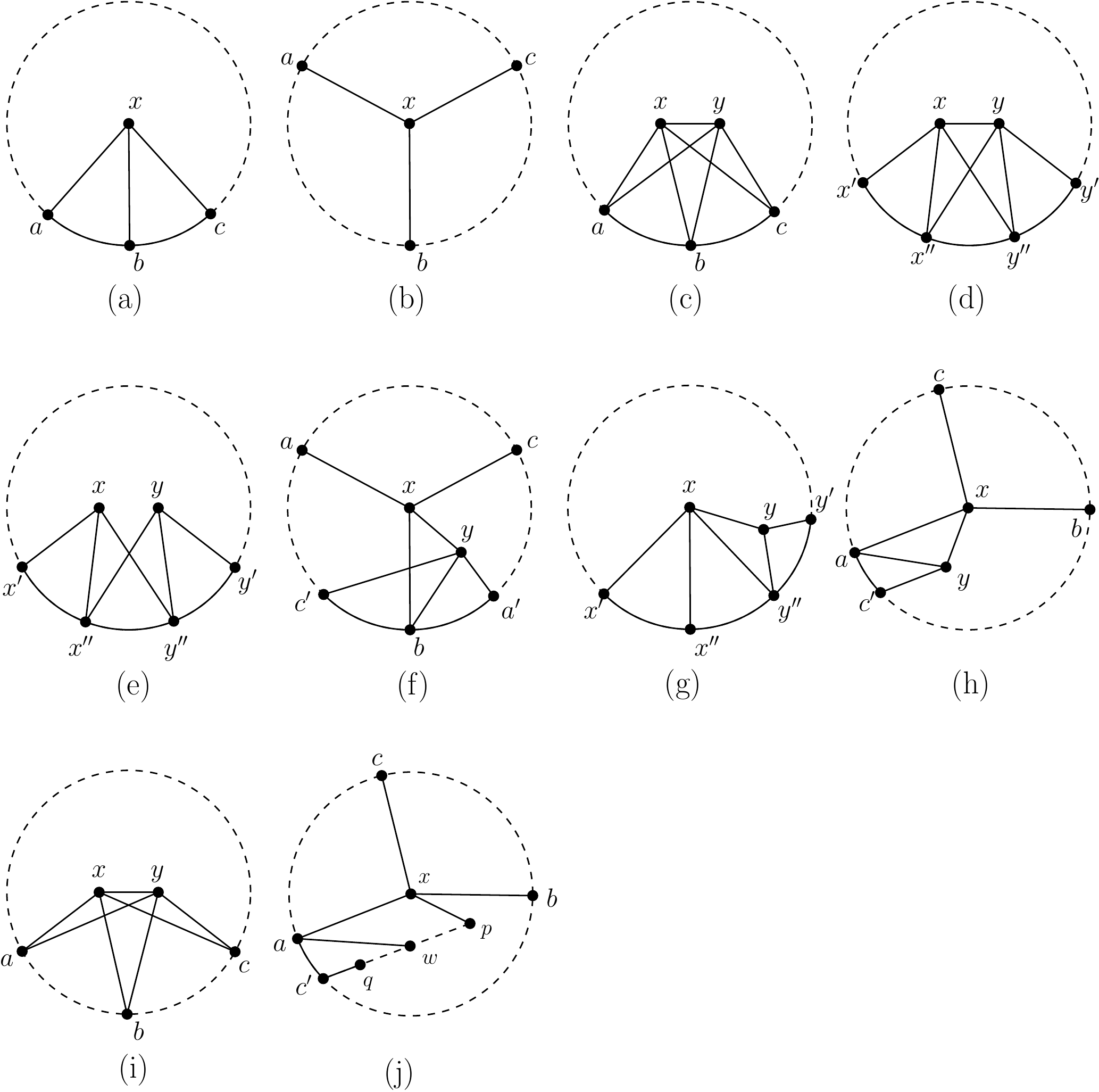}
	\caption{Patterns in the pyramid-free case (solid lines represent edges)}
	\label{fig:wheel-family}
\end{figure}

\begin{figure}[ht]
	\centering \includegraphics[scale=0.5]{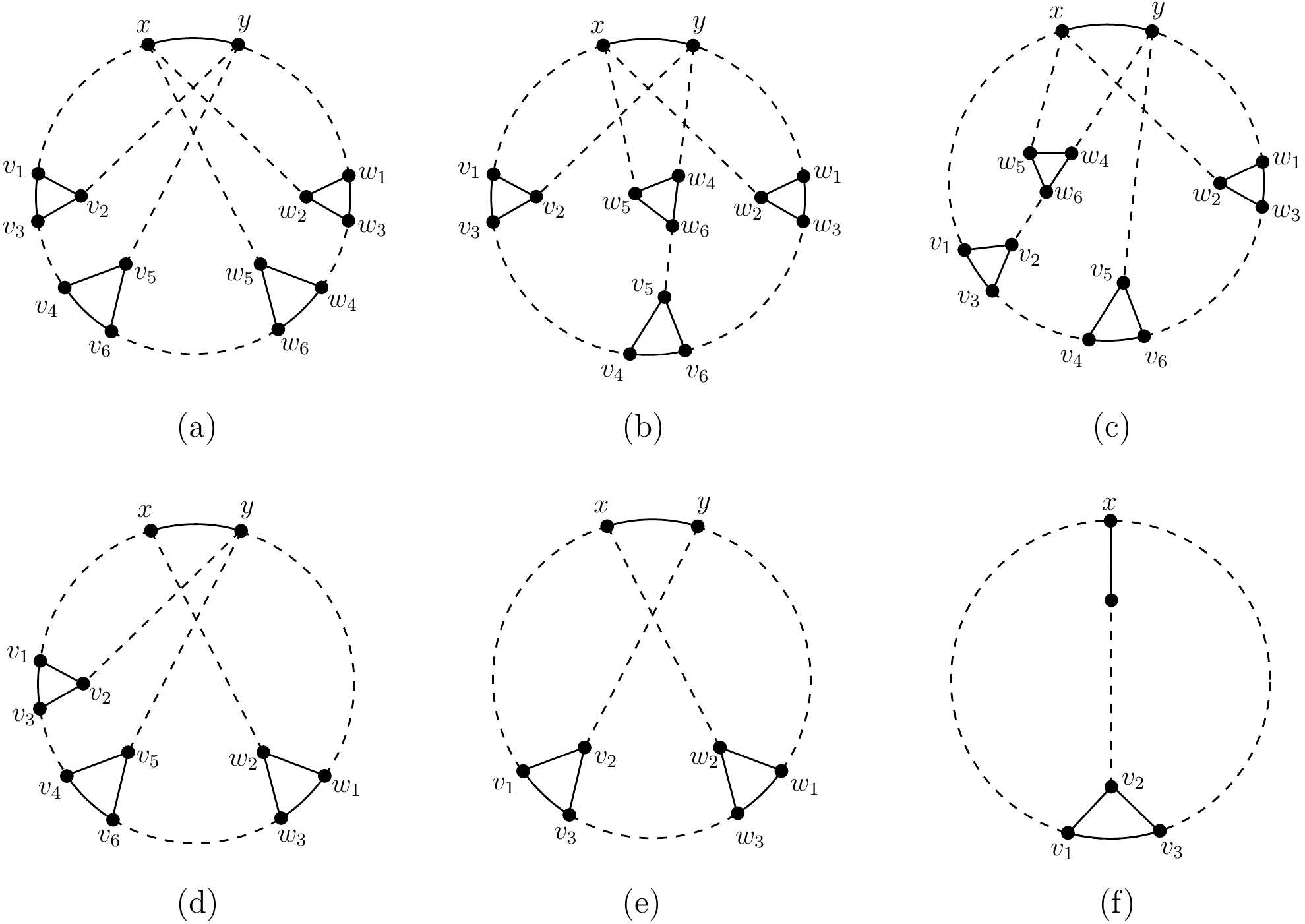}
	\caption{Patterns that contain pyramid}
	\label{fig:pattern-pyramid}
\end{figure}

A {\em 2-join} in a graph $G$ is a partition of $V(G)$ into two sets $X_1$, $X_2$ each of size at least~3, such that for $i=1, 2$, $X_i$ contains two non-empty disjoint sets $A_i$, $B_i$, $A_1$ is complete to $A_2$, $B_1$ is complete to $B_2$, and there are no other edges between $X_1$ and $X_2$. Moreover, for $i=1, 2$, $X_i$ does not consist of a path with one end in $A_i$, one end in $B_i$ and no internal vertex in $A_i\cup B_i$.

We are not sure that our list of patterns is complete for our class, but we believe that the real list is close to it and, above all, finite.  This should imply that the tree-width is bounded.  Also, we wonder whether a similar approach can be extended to even-hole-free graphs of maximum degree~$k$ for any fixed integer $k$. Observe that for $k=3$, this is what we actually do in Theorem~\ref{th:main-thm-max3}, since the list of basic graphs can be seen as obtained by a finite list of patterns and the so-called proper separator is a special case of 2-join.  For $k\geq 5$, rings (already defined in Section~\ref{sec:delta4}) become a problem, but an extension of the notion of 2-join might lead to a true statement.

%% file: Main.bbl
\begin{thebibliography}{10}

\bibitem{AdlerH18}
Isolde Adler and Frederik Harwath.
\newblock Property testing for bounded degree databases.
\newblock In Rolf Niedermeier and Brigitte Vall{\'{e}}e, editors, {\em 35th
  Symposium on Theoretical Aspects of Computer Science, {STACS} 2018, February
  28 to March 3, 2018, Caen, France}, volume~96 of {\em LIPIcs}, pages
  6:1--6:14. Schloss Dagstuhl - Leibniz-Zentrum fuer Informatik, 2018.

\bibitem{BeinekeS75}
Lowell Beineke and Allen Schwenk.
\newblock On a bipartite form of the ramsey problem.
\newblock {\em Proc. 5th British Combinatorial Conference 1975, Congressus
  Numer.~{XV}}, pages 17--22, 1975.

\bibitem{benjamini2010every}
Itai Benjamini, Oded Schramm, and Asaf Shapira.
\newblock Every minor-closed property of sparse graphs is testable.
\newblock {\em Advances in mathematics}, 223(6):2200--2218, 2010.

\bibitem{BodlaenderT97}
Hans~L. Bodlaender and Dimitrios~M. Thilikos.
\newblock Treewidth for graphs with small chordality.
\newblock {\em Discrete Applied Mathematics}, 79(1-3):45--61, 1997.

\bibitem{DBLP:journals/jgt/BoncompagniPV19}
Valerio Boncompagni, Irena Penev, and Kristina Vuskovic.
\newblock Clique-cutsets beyond chordal graphs.
\newblock {\em Journal of Graph Theory}, 91(2):192--246, 2019.

\bibitem{CameronCH18}
Kathie Cameron, Steven Chaplick, and Ch{\'{\i}}nh~T. Ho{\`{a}}ng.
\newblock On the structure of (pan, even hole)-free graphs.
\newblock {\em Journal of Graph Theory}, 87(1):108--129, 2018.

\bibitem{CameronSHV18}
Kathie Cameron, Murilo V.~G. da~Silva, Shenwei Huang, and Kristina Vuskovic.
\newblock Structure and algorithms for (cap, even hole)-free graphs.
\newblock {\em Discret. Math.}, 341(2):463--473, 2018.

\bibitem{DBLP:journals/corr/abs-1912-11246}
Maria Chudnovsky, St{\'{e}}phan Thomass{\'{e}}, Nicolas Trotignon, and Kristina
  Vu{\v s}kovi{\' c}.
\newblock Maximum independent sets in (pyramid, even hole)-free graphs.
\newblock {\em CoRR}, abs/1912.11246, 2019.

\bibitem{DBLP:journals/iandc/Courcelle90}
Bruno Courcelle.
\newblock The monadic second-order logic of graphs. i. recognizable sets of
  finite graphs.
\newblock {\em Inf. Comput.}, 85(1):12--75, 1990.

\bibitem{CourcelleE12}
Bruno Courcelle and Joost Engelfriet.
\newblock {\em Graph Structure and Monadic Second-Order Logic - {A}
  Language-Theoretic Approach}, volume 138 of {\em Encyclopedia of mathematics
  and its applications}.
\newblock Cambridge University Press, 2012.

\bibitem{FominGT11}
Fedor~V. Fomin, Petr~A. Golovach, and Dimitrios~M. Thilikos.
\newblock Contraction obstructions for treewidth.
\newblock {\em J. Comb. Theory, Ser. {B}}, 101(5):302--314, 2011.

\bibitem{forster2019computing}
Sebastian Forster, Danupon Nanongkai, Thatchaphol Saranurak, Liu Yang, and
  Sorrachai Yingchareonthawornchai.
\newblock Computing and testing small connectivity in near-linear time and
  queries via fast local cut algorithms.
\newblock {\em SODA}, 2020.

\bibitem{goldreich2017introduction}
Oded Goldreich.
\newblock {\em Introduction to property testing}.
\newblock Cambridge University Press, 2017.

\bibitem{GoldreichRon2002}
Oded Goldreich and Dana Ron.
\newblock Property testing in bounded degree graphs.
\newblock {\em Algorithmica}, 32(2):302--343, 2002.

\bibitem{hassidim2009local}
Avinatan Hassidim, Jonathan~A Kelner, Huy~N Nguyen, and Krzysztof Onak.
\newblock Local graph partitions for approximation and testing.
\newblock In {\em 2009 50th Annual IEEE Symposium on Foundations of Computer
  Science}, pages 22--31. IEEE, 2009.

\bibitem{AdlerKoehler2020}
Noleen~K\"{o}hler Isolde~Adler.
\newblock On graphs of bounded degree that are far from being {H}amiltonian.
\newblock Submitted.

\bibitem{kawarabayashi2013testing}
Ken-ichi Kawarabayashi and Yuichi Yoshida.
\newblock Testing subdivision-freeness: property testing meets structural graph
  theory.
\newblock In {\em Proceedings of the forty-fifth annual ACM symposium on Theory
  of computing}, pages 437--446. ACM, 2013.

\bibitem{kumar2019random}
Akash Kumar, C~Seshadhri, and Andrew Stolman.
\newblock Random walks and forbidden minors ii: a poly (d
  $\varepsilon$-1)-query tester for minor-closed properties of bounded degree
  graphs.
\newblock In {\em Proceedings of the 51st Annual ACM SIGACT Symposium on Theory
  of Computing}, pages 559--567, 2019.

\bibitem{NewmanSohler2013}
Ilan Newman and Christian Sohler.
\newblock Every property of hyperfinite graphs is testable.
\newblock {\em SIAM Journal on Computing}, 42(3):1095--1112, 2013.

\bibitem{RobertsonS86}
Neil Robertson and Paul~D. Seymour.
\newblock Graph minors. v. excluding a planar graph.
\newblock {\em J. Comb. Theory, Ser. {B}}, 41(1):92--114, 1986.

\bibitem{SilvaSS10}
Ana Silva, Aline~Alves da~Silva, and Cl{\'{a}}udia~Linhares Sales.
\newblock A bound on the treewidth of planar even-hole-free graphs.
\newblock {\em Discrete Applied Mathematics}, 158(12):1229--1239, 2010.

\bibitem{DBLP:journals/corr/abs-1906-10998}
Ni~Luh~Dewi Sintiari and Nicolas Trotignon.
\newblock (theta, triangle)-free and (even hole, k\({}_{\mbox{4}}\))-free
  graphs. part 1 : Layered wheels.
\newblock {\em CoRR}, abs/1906.10998, 2019.

\bibitem{Thomason82}
Andrew Thomason.
\newblock On finite ramsey numbers.
\newblock {\em Eur. J. Comb.}, 3(3):263--273, 1982.

\bibitem{vuskovic:evensurvey}
Kristina Vuskovic.
\newblock Even-hole-free graphs: A survey.
\newblock {\em Applicable Analysis and Discrete Mathematics}, 4, 10 2010.

\bibitem{yoshida2012property}
Yuichi Yoshida and Hiro Ito.
\newblock Property testing on $k$-vertex-connectivity of graphs.
\newblock {\em Algorithmica}, 62(3-4):701--712, 2012.

\end{thebibliography}
